\documentclass[12pt, onecolumn, draftcls]{IEEEtran}
\newlength{\figwidth}
\ifCLASSOPTIONonecolumn
\else
\fi
\usepackage[dvips]{graphicx}
\usepackage{subfigure}
\usepackage{epsfig, psfrag, amsmath, amssymb, amsfonts, latexsym, balance, hyperref,indentfirst}
\usepackage[numbers,sort&compress]{natbib}
\usepackage{algorithm, algorithmic}
\usepackage{multirow}
\usepackage{amsthm}
\usepackage{xcolor}
\usepackage{extarrows}
\usepackage[small]{caption2}

\newtheorem{myDef}{Definition}
\newtheorem{myTheo}{Theorem}

\newtheorem{cor}{Corollary}

\newtheorem{prop}{Proposition}

\begin{document}
\setcounter{page}{1}
\title{The Ginibre Point Process as a Model for Wireless Networks with Repulsion}
\author{Na Deng, Wuyang Zhou, \emph{Member, IEEE}, and Martin Haenggi, \emph{Fellow, IEEE}
\thanks{Na Deng and Wuyang Zhou are with the Dept. of Electronic Engineering and Information Science, University of Science and Technology of China (USTC), Hefei, 230027, China (e-mail: ndeng@mail.ustc.edu.cn, wyzhou@ustc.edu.cn).
Martin Haenggi is with the Dept. of Electrical Engineering, University of Notre Dame, Notre Dame 46556, USA (e-mail: mhaenggi@nd.edu).

This work was supported by National programs for High Technology Research and Development (2012AA011402), National Natural Science Foundation of China (61172088), and by the US NSF grants CCF 1216407 and CNS 1016742.
}}

\maketitle
\thispagestyle{empty}
\begin{abstract}
The spatial structure of transmitters in wireless networks plays a key role in evaluating the mutual interference and hence the performance.
Although the Poisson point process (PPP) has been widely used to model the spatial configuration of wireless networks, it is not suitable for networks with repulsion. The Ginibre point process (GPP) is one of the main examples
of determinantal point processes that can be used to model random phenomena where repulsion is observed. Considering the
accuracy, tractability and practicability tradeoffs, we introduce and promote the $\beta$-GPP, an intermediate class
between the PPP and the GPP, as a model for wireless networks when the nodes
exhibit repulsion. To show that the model leads to analytically tractable results in several cases of interest, we derive the mean and variance
of the interference using two different approaches: the Palm measure approach and the reduced second moment approach, and then provide approximations of the interference distribution by three known probability density functions. Besides, to show that the model is relevant for cellular systems, we derive the coverage probability of the typical user and also find that the fitted $\beta$-GPP can closely model
the deployment of actual base stations in terms of the coverage probability and other statistics.
\end{abstract}

\section{Introduction}
\label{sec:introduction}
\subsection{Motivation}
\label{subsec:motivation}
The spatial distribution of transmitters is critical in determining the power of the received signals and the mutual interference, and hence the performance, of a wireless network. As a consequence, stochastic geometry models for wireless networks have recently received widespread attention since they capture the irregularity and variability of the node configurations found in real networks and provide theoretical insights (see, e.g., \cite{haenggi2009stochastic,andrews2010primer,baccelli2010stochastic,haenggi2012stochastic}).
However, most works that use such stochastic geometry models for wireless networks assume, due to its tractability, that nodes are distributed according to the Poisson point process (PPP), which means the wireless nodes are located independently of each other. However, in many actual wireless networks, node locations are spatially correlated, i.e., there exists repulsion (or attraction) between nodes, which means that the node locations of a real deployment usually appear to form a more regular (or more clustered) point pattern than the PPP. In this paper, we focus on networks that are more regular than the PPP. In order to model such networks more accurately, {\em hard-core} or {\em soft-core} processes that account for the repulsion between nodes are required.

It is well recognized that point processes like the Mat\'ern hard-core process (MHCP), the Strauss process and the perturbed lattice are more realistic point process models than the PPP and the hexagonal grid models since they can capture the spatial characteristics of the actual network deployments better \cite{haenggi2012stochastic}. More importantly, through fitting, these point processes may have nearly the same characteristics as the real deployment of network nodes \cite{guo2013spatial}. However, the main problem with these point processes is their limited analytical tractability, which makes it difficult to analyze the properties of these repulsive point processes, thus limiting their applications in wireless networks.

In this paper, we propose the {\em Ginibre point process} (GPP) as a model for wireless networks whose nodes exhibit repulsion. The GPP belongs to the class of determinantal point processes (DPPs) \cite{lavancier2013determinantal} and thus is a soft-core model. It is less regular than the lattice but more regular than the PPP. Specifically, we focus on a more general point process denoted as $\beta$-GPP, $0<\beta\leq1$, which is a thinned and re-scaled GPP. It is the point process obtained by retaining, independently and with probability $\beta$, each point of the GPP and then applying the homothety of ratio $\sqrt{\beta}$ to the remaining points in order to maintain the original intensity of the GPP. Note that the $1$-GPP is the GPP and that the $\beta$-GPP converges weakly to the PPP of the same intensity as $\beta\rightarrow0$ \cite{torrisi2013large}. In other words, the family of $\beta$-GPPs generalizes the GPP and also includes the PPP as a limiting case. Intuitively, we can always find a $\beta$ of the $\beta$-GPP to model those real network deployments with repulsion accurately as long as they are not more regular than the GPP. In view of this, we fit the $\beta$-GPPs to the base station (BS) locations in real cellular networks obtained from a public database. Another attractive feature of the GPP is that it has some critical properties (in particular, the form of its moment densities and its Palm distribution) that other soft-core processes do not share, which enables us to obtain expressions or bounds of important performance metrics in wireless networks.

\subsection{Contributions}
\label{subsec:contribution}
The main objective of this paper is to introduce and promote the GPP as a model for wireless networks where nodes exhibit repulsion. The GPP not only captures the geometric characteristics of real networks but also is fairly tractable analytically, in contrast to other repulsive point process models.

Specifically, we first derive some classical statistics, such as the $K$ function, the $L$ function and the $J$ function, the distribution of each point's modulus and the Palm measure for the $\beta$-GPPs. And then, to show that the model leads to tractable results in several cases of interest, we analyze the mean and variance of the interference using two different approaches: one is with the aid of the Palm measure and the other is to use the reduced second moment measure. We also provide comparisons of the mean interference between the $\beta$-GPPs and other common point processes. Further, based on the mean and variance of the interference, we provide approximations of the interference distribution using three known probability density functions, the gamma distribution, the inverse Gaussian distribution, and the inverse gamma distribution.

As an application to cellular systems, we derive a computable representation for the coverage probability in cellular networks---the probability that the signal-to-interference-plus-noise ratio (SINR) for a mobile user achieves a target threshold. To show that the model is indeed relevant for cellular systems, we fit the $\beta$-GPPs to the BS locations in real cellular networks obtained from a public database. The fitting results demonstrate that the fitted $\beta$-GPP has nearly the same coverage properties as the given point set and thus, in terms of coverage probability, it can closely model the deployment of actual BSs. We also find that the fitted value of $\beta$ is close to 1 for urban regions and relatively small (between 0.2 and 0.4) for rural regions, which means the nodes in actual cellular networks are less regular than the GPP and thus can always be modeled by a $\beta$-GPP through tuning the parameter $\beta$. Therefore, the $\beta$-GPP is a very useful and accurate point process model for most repulsive wireless networks, especially for cellular networks.

\subsection{Related Work}
\label{subsec:related work}
Stochastic geometry models have been successfully applied to model and analyze wireless networks in the last two decades since they not only capture the topological randomness in the network geometry but also lead to tractable analytical results \cite{haenggi2009stochastic}. The PPP is by far the most popular point process used in the literature because of its tractability. Models based on the PPP have been used for a variety of networks, including cellular networks, mobile ad hoc networks, cognitive radio networks, and wireless sensor networks, and the performance of PPP-based networks is well characterized and well understood (see, e.g., \cite{andrews2010primer,haenggi2009stochastic,baccelli2010stochastic,elsawy2013stochastic} and references therein). Although the PPP model provides many useful theoretical results, the independence of the node locations makes the PPP a dubious model for actual network deployments where the wireless nodes appear spatially negatively correlated, i.e., where the nodes exhibit repulsion. Hence, point processes that consider the spatial correlations, such as the MHCP and the Strauss process, have been explored recently since they can better capture the spatial distribution of the network nodes in real deployments \cite{haenggi2012stochastic,guo2013spatial,haenggi2011mean}. However, their limited tractability impedes further applications in wireless networks and leaves many challenges to be addressed.

The GPP, one of the main examples of determinantal point processes on the complex plane, has recently been proposed as a model for cellular networks in the technical report \cite{miyoshi2012gpp}. Although there is a vast body of research on GPPs used to model random phenomena with repulsion, no refereed article has used the GPP as a model for wireless networks. Only two (non-refereed) works focus on this model: \cite{miyoshi2012gpp} proposes a stochastic geometry model of cellular networks according to the GPP and derives the coverage probability and its asymptotics as the SINR threshold becomes large but does not consider the more general case of the $\beta$-GPP; \cite{torrisi2013large} investigates the asymptotic behavior of the tail of the interference using different fading random variables in wireless networks, where nodes are distributed according to the $\beta$-GPP. Different from them, we mainly focus on the mean and variance of the interference in $\beta$-Ginibre wireless networks, the coverage, and the fitting to real data. Since the family of $\beta$-GPPs constitutes an intermediate class between the PPP (fully random) and the GPP (relatively regular), it is intuitive that we can use the $\beta$-GPP to model a large class of actual wireless networks by tuning the parameter $\beta$.

\subsection{Mathematical Preliminaries}
\label{subsec: Mathematical Preliminaries}
Here we give a brief overview of some terminology, definitions and properties for the Ginibre point process. Readers are referred to \cite{lavancier2013determinantal,miyoshi2012gpp,decreusefond2013efficient,goldman2010palm} for further details. We use the following notation. Let $\mathbb{C}$ denote the complex plane. For a complex number $z=z_1+jz_2$ (where $z_1,z_2\in\mathbb{R}$ and $j=\sqrt{-1}$), we denote by $\bar{z}=z_1-jz_2$ the complex conjugate and by $|z|=\sqrt{z_1^2+z_2^2}$ the modulus. Consider a Borel set $S\subseteq\mathbb{R}^d$. For any function $K: S\times S\rightarrow\mathbb{C}$, let $[K](x_1,\ldots,x_n)$ be the $n\times n$ matrix with $(i,j)$'th entry $K(x_i,x_j)$. For a square complex matrix $A$, let $\det A$ denote its determinant.
\begin{myDef}
{\rm\textbf{(Determinantal point processes).}} Suppose that a simple locally finite spatial point process $\Phi$ on $S\subseteq\mathbb{R}^d$ has product density functions
\ifCLASSOPTIONonecolumn
\begin{equation} \label{6}
\varrho^{(n)}(x_1,\ldots,x_n)=\det[K](x_1,\ldots,x_n),~~~(x_1,\ldots,x_n)\in S^{n},~~n=1,2,\ldots,
\end{equation}
\else
\begin{multline} \label{6}
\varrho^{(n)}(x_1,\ldots,x_n)=\det[K](x_1,\ldots,x_n), \\
(x_1,\ldots,x_n)\in S^{n},~~n=1,2,\ldots,
\end{multline}
\fi
and for any Borel function $h$: $S^n\rightarrow [0,\infty)$,
\ifCLASSOPTIONonecolumn
\begin{equation} \label{2}
\mathbb{E}\sum\limits_{x_1,\ldots,x_n\in\Phi}^{\neq}h(x_1,\ldots,x_n)=\int_{S}\cdots\int_{S}\varrho^{(n)}(x_1,\ldots,x_n)h(x_1,\ldots,x_n)dx_1\cdots dx_n,
\end{equation}
\else
\begin{multline} \label{2}
\mathbb{E}\sum\limits_{x_1,\ldots,x_n\in\Phi}^{\neq}h(x_1,\ldots,x_n) \\
=\int_{S}\cdots\int_{S}\varrho^{(n)}(x_1,\ldots,x_n)h(x_1,\ldots,x_n)dx_1\cdots dx_n,
\end{multline}
\fi
Then $\Phi$ is called a determinantal point process (DPP) with kernel $K$, and we write $\Phi\sim$ DPP$_S(K)$.
\end{myDef}
This definition implies that $\varrho^{(n)}(x_1,...,x_n)=0$ when $x_i=x_j$ for $i\neq j$, which reflects the repulsiveness of a DPP. $\varrho\equiv\varrho^{(1)}$ is the intensity function and $g(x,y)=\varrho^{(2)}(x,y)/[\varrho(x)\varrho(y)]$ is the pair correlation function, where we set $g(x,y)=0$ if $\varrho(x)$ or $\varrho(y)$ is zero.
\begin{myDef}
{\rm\textbf{(The standard Ginibre point process).}} $\Phi$ is said to be the Ginibre point process (GPP) when the kernel $K$ is given by
\begin{equation}
K(x,y)=e^{x\bar{y}},~~x,y\in\mathbb{C},
\end{equation}
with respect to the Gaussian measure $\nu(dz)=\pi^{-1}e^{-|z|^2}m(dz)$, where $m$ denotes the Lebesgue measure on $\mathbb{C}$. Alternatively,
\begin{equation}
K(x,y)=\pi^{-1}e^{-(|x|^2+|y|^2)/2}e^{x\bar{y}},~~x,y\in\mathbb{C},
\end{equation}
with respect to the Lebesgue measure $\nu(dz)=m(dz)$ \cite{miyoshi2012gpp}.
\end{myDef}
\emph{Remarks:}
\begin{enumerate}
  \item By the definition of the GPP and from (\ref{2}), we see that $\mathbb{E}[\Phi(S)]=\int_S \varrho(x)dx=\pi^{-1}|S|$, $S\subset\mathbb{C}$; that is, the first-order density of the GPP is $\pi^{-1}$.
  \item Since the GPP is motion-invariant, the second moment density depends only on the distance of its arguments, i.e., $\varrho^{(2)}(x,y)=\varrho^{(2)}(r)$, where $r=|x-y|$, $\forall x,y\in\mathbb{C}$. We have
\ifCLASSOPTIONonecolumn
\begin{equation}
\!\!\!\!\!\varrho^{(2)}(x,y)\!=\!\det\!\begin{bmatrix}
   \pi^{-1}e^{-(|x|^2+|x|^2)/2}e^{x\bar{x}}\!&\!\pi^{-1}e^{-(|x|^2+|y|^2)/2}e^{x\bar{y}}\\
   \pi^{-1}e^{-c(|x|^2+|y|^2)/2}e^{\bar{x}y}\!&\!\pi^{-1}e^{-(|y|^2+|y|^2)/2}e^{y\bar{y}}\\
\end{bmatrix}
\!=\!\pi^{-2}(1-e^{-r^2})=\varrho^{(2)}(r).
\end{equation}
\else
\begin{multline}
\!\!\!\!\!\varrho^{(2)}(x,y) \\
\!\!\!\!\!=\det\!\begin{bmatrix}
   \pi^{-1}e^{-(|x|^2+|x|^2)/2}e^{x\bar{x}}\!&\!\pi^{-1}e^{-(|x|^2+|y|^2)/2}e^{x\bar{y}}\\
   \pi^{-1}e^{-c(|x|^2+|y|^2)/2}e^{\bar{x}y}\!&\!\pi^{-1}e^{-(|y|^2+|y|^2)/2}e^{y\bar{y}}\\
\end{bmatrix} \\
=\pi^{-2}(1-e^{-r^2})=\varrho^{(2)}(r).
\end{multline}
\fi
  \item The pair correlation function \cite[Def.~6.6]{haenggi2012stochastic} is
\begin{equation}
g(x,y)\triangleq\frac{\varrho^{(2)}(x,y)}{\varrho(x)\varrho(y)}=\frac{\pi^{-2}(1-e^{-r^2})}{\pi^{-2}}=1-e^{-r^2}.
\end{equation}
Since $g(r)<1~\forall r$, the GPP is repulsive at all distances.
\end{enumerate}
\begin{myDef}
{\rm\textbf{(The thinned and re-scaled Ginibre point process).}} The thinned and re-scaled Ginibre point process ($\beta$-GPP), $0<\beta\leq1$, is a point process obtained by retaining, independently and with probability $\beta$, each point of the GPP and then applying the homothety of ratio $\sqrt{\beta}$ to the remaining points in order to maintain the original intensity of the GPP.
\end{myDef}
Note that the 1-GPP is the GPP and that the $\beta$-GPP converges weakly to the PPP of intensity $1/\pi$ as $\beta\rightarrow0$. The parameter $\beta$ can be used to ``interpolate'' smoothly from the GPP to the PPP. And the $\beta$-GPPs are still determinantal processes and satisfy the usual conditions of existence and uniqueness (see, e.g., \cite{goldman2010palm}).
\begin{myDef} \label{3}
{\rm\textbf{(The scaled $\beta$-Ginibre point process).}} The scaled $\beta$-Ginibre point process is a scaled version of the $\beta$-GPP on the complex plane with intensity $\lambda=c/\pi$, where $c>0$ is the scaling parameter used to control the intensity. The kernel of the scaled $\beta$-GPP is given by $K_{\beta,c}(x,y)=\beta e^{\frac{c}{\beta}x\bar{y}}$ with respect to the reference measure $\nu_{\beta,c}(dz)=\frac{c}{\beta\pi}e^{-\frac{c}{\beta}|z|^2}m(dz)$, or, equivalently, $K_{\beta,c}(x,y)=(c/\pi)e^{-\frac{c}{2\beta}(|x|^2+|y|^2)}e^{\frac{c}{\beta}x\bar{y}}$ with respect to the Lebesgue measure \cite{goldman2010palm}.
\end{myDef}

\subsection{Organization}
\label{subsec:organization}
The rest of the paper is organized as follows: Section \ref{sec:properties of the beta-GPP} introduces some basic but important properties of the $\beta$-GPPs. Section \ref{sec:Analysis of the mean and variance of interference} analyzes the mean and variance of the interference and make approximations of the interference distribution. Section \ref{sec:Coverage probability} derives the coverage probability of the typical user in $\beta$-Ginibre wireless networks. Section \ref{sec:Fitting the beta-GPP to actual network deployments} presents fitting results of the $\beta$-GPPs to actual network deployments, and Section \ref{sec:conclusions} offers the concluding remarks.

\section{Properties of the $\beta$-Ginibre Point Processes}
\label{sec:properties of the beta-GPP}
\subsection{Basic properties}
\label{subsec:basic properties}
In this section, we consider a scaled $\beta$-GPP, denoted by $\Phi_c$. Since the GPP is motion-invariant, the second moment density $\varrho_{\beta,c}^{(2)}(x,y)=\varrho_{\beta,c}^{(2)}(u)$, where $u=|x-y|,~\forall x,y\in\mathbb{C}$, and the third moment density $\varrho_{\beta,c}^{(3)}(x,y,z)=\varrho_{\beta,c}^{(3)}(o,w_1,w_2)$, where $w_1=y-x$, $w_2=z-x$, $\forall x,y,z\in\mathbb{C}$. Based on Definition \ref{3} and (\ref{6}), we have
\begin{equation}
\varrho_{\beta,c}^{(2)}(u)\!=\!\frac{c^2}{\pi^2}(1-e^{-\frac{c}{\beta}u^2}),
\end{equation}
\ifCLASSOPTIONonecolumn
\begin{equation} \label{8}
\varrho_{\beta,c}^{(3)}(o,w_1,w_2)=\frac{c^3}{\pi^{3}}\Big(1\!-\!e^{-\frac{c}{\beta}|w_1|^2}\!\!-\!e^{-\frac{c}{\beta}|w_2|^2}\!\!-\!e^{-\frac{c}{\beta}|w_1-w_2|^2}\big(1-e^{-\frac{c}{\beta}w_1\bar{w_2}}\!\!-\!e^{-\frac{c}{\beta}w_2\bar{w_1}}\big)\Big).
\end{equation}
\else
\begin{multline} \label{8}
\varrho_{\beta,c}^{(3)}(o,w_1,w_2)=\frac{c^3}{\pi^{3}}\Big(1-e^{-\frac{c}{\beta}|w_1|^2}-e^{-\frac{c}{\beta}|w_2|^2} \\
-e^{-\frac{c}{\beta}|w_1-w_2|^2}\big(1-e^{-\frac{c}{\beta}w_1\bar{w_2}}-e^{-\frac{c}{\beta}w_2\bar{w_1}}\big)\Big).
\end{multline}
\fi

It is known that the moduli (on the complex plane) of the points of the GPP have the same distribution as independent gamma random variables \cite{kostlan1992spectra}. For the $\beta$-GPP, from Theorem 4.7.1 in \cite{hough2009zeros}, we have the following result:
\begin{prop} \label{4}
Let $\Phi_c=\{X_i\}_{i\in\mathbb{N}}$ be a scaled $\beta$-GPP. For $k\in\mathbb{N}$, let $Q_k$ be a random variable with probability density function
\begin{equation} \label{10}
f_{Q_k}(q)=\frac{q^{k-1}e^{-\frac{c}{\beta}q}}{(\beta/c)^k\Gamma(k)},
\end{equation}
i.e., $Q_k\sim {\rm gamma}(k,\beta/c)$, with $Q_i$ independent of $Q_j$ if $i\neq j$.
Then the set $\{|X_i|^2\}_{i\in\mathbb{N}}$ has the same distribution as the set $\Xi$ obtained by retaining from $\{Q_k\}_{k\in\mathbb{N}}$ each $Q_k$ with probability $\beta$ independently of everything else.
\end{prop}
From Theorem 1 and Remark 24 in \cite{goldman2010palm}, we know that there exists a version of the GPP $\Phi_c$ such that the Palm measure of $\Phi_c$ is the law of the process obtained by removing from $\Phi_c$ a Gaussian-distributed point and then adding the origin. Thus, we have the following proposition.
\begin{prop} \label{5}
{\rm\textbf{(The Palm measure of the scaled $\beta$-Ginibre point process).}} For a scaled $\beta$-GPP $\Phi_c$, the Palm measure of $\Phi_c$ is the law of the process obtained by adding the origin and deleting the point $X$ if it belongs (which occurs with probability $\beta$) to the process $\Phi_c$, where $|X|^2=Q_1$.
\end{prop}
From Propositions \ref{4} and \ref{5}, we observe that the Palm distribution of the squared moduli $Q_k$ is closely related to the non-Palm version, the only difference being that $Q_1$ is removed if it is included in $\Xi$.

\subsection{$K$ and $L$ functions for the $\beta$-Ginibre point processes}
\label{subsec:K function of beta-GPP}
\subsubsection{$K$ function}
\label{subsubsec:K function}
The $K$ function is defined as $K(r)\triangleq\frac{1}{\lambda}\mathcal{K}(b(o,r))$ \cite{ripley1976second}, where $\mathcal{K}$ is the reduced second moment measure of the point process, given by \cite{haenggi2012stochastic}
\begin{equation} \label{7}
\mathcal{K}(B)=\frac{1}{\lambda}\int_{B}\varrho^{(2)}(u)du,
\end{equation}
and $b(o,r)$ is the ball of radius $r$ centered at the origin $o$, so $K'(r)dr=\frac{2\pi}{\lambda}\mathcal{K}(rdr)$. Therefore,
\ifCLASSOPTIONonecolumn
\begin{equation}
K(r)=\frac{2\pi}{\lambda^2}\int_0^r\frac{c^2}{\pi^2}(1-e^{-\frac{c}{\beta}u^2})udu=\pi r^2-\frac{\beta\pi}{c}(1-e^{-\frac{c}{\beta}r^2}).
\end{equation}
\else
\begin{eqnarray}
K(r)\!\!\!\!&=&\!\!\!\!\frac{2\pi}{\lambda^2}\int_0^r\frac{c^2}{\pi^2}(1-e^{-\frac{c}{\beta}u^2})udu \nonumber \\
\!\!\!\!&=&\!\!\!\!\pi r^2-\frac{\beta\pi}{c}(1-e^{-\frac{c}{\beta}r^2}).
\end{eqnarray}
\fi
It is easily verified that $K(r)\rightarrow \pi r^2$ as $\beta\rightarrow 0$, which is the $K$ function of the PPPs.

Although the points of a GPP exhibit repulsion, there is no hard restriction about the distance between any two points; therefore the GPP is a soft-core process. In contrast, for a hard-core process, points are strictly forbidden to be closer than a certain minimum distance. In order to compare the properties of hard-core and the GPP, we also present the $K$ functions for the Mat\'ern hard-core processes (MHCP) of type I and type II. The $K$ function for the MHCP of type I with minimum distance $\delta$ has been given by \cite{haenggi2011mean}
\begin{equation} \label{14}
K_{\rm I}(r)=2\pi\frac{\lambda_p^2}{\lambda_{\rm I}^2}\int_0^ruk_{\rm I}(u)du,
\end{equation}
where
\begin{equation}
k_{\rm I}(u)=\left\{ \begin{matrix}
   {0, ~~~~~~~~~~~~~~~~\quad u<\delta}  \\
   {\exp(-\lambda_pV_{\delta}(u)), \quad u\geq\delta}  \\
\end{matrix} \right.
\end{equation}
is the probability that two points at distance $u$ are both retained, $\lambda_p$ is the intensity of the stationary parent PPP,
and $\lambda_{\rm I}=\lambda_p\exp(-\lambda_p\pi\delta^2)$. $V_{\delta}(u)$ is the area of the union of two disks of radius $\delta$ whose centers are separated
by $u$, given by
\begin{equation}
\!V_{\delta}(u)\!=\!2\pi\delta^2\!-\!2\delta^2\arccos\left(\frac{u}{2\delta}\right)\!+\!u\sqrt{\delta^2\!-\!\frac{u^2}{4}},~0\leq u\leq 2\delta.
\end{equation}
For $u>2\delta$, $V_{\delta}(u)=2\pi\delta^2$. When $r>2\delta$ and $\lambda_{\rm I}=\lambda$,
\ifCLASSOPTIONonecolumn
\begin{equation}
\lim\limits_{r\rightarrow\infty}K_{\rm GPP}(r)-K_{\rm I}(r)=4\pi\delta^2-\frac{1}{\lambda_p\exp(-\pi\lambda_p\delta^2)}-2\pi\frac{\lambda_p^2}{\lambda^2}\int_{\delta}^{2\delta}k_{\rm I}(u)udu.
\end{equation}
\else
\begin{multline} \label{8}
\lim\limits_{r\rightarrow\infty}K_{\rm GPP}(r)-K_{\rm I}(r) \\
=4\pi\delta^2-\frac{1}{\lambda_p\exp(-\pi\lambda_p\delta^2)}-2\pi\frac{\lambda_p^2}{\lambda^2}\int_{\delta}^{2\delta}k_{\rm I}(u)udu.
\end{multline}
\fi
Since $\frac{\lambda_p^2}{\lambda^2}k_{\rm I}(u)$ is monotonically increasing in $\lambda_p$ for all $\delta\leq u <2\delta$, $\frac{\lambda_p^2}{\lambda^2}k_{\rm I}(u)\geq 1$. Thus,
\begin{equation} \label{17}
\lim\limits_{r\rightarrow\infty}K_{\rm GPP}(r)-K_{\rm I}(r)\leq\frac{\pi\lambda_p\delta^2-\exp(\pi\lambda_p\delta^2)}{\lambda_p}<0.
\end{equation}

The $K$ function for the MHCP of type II can be expressed as \cite{haenggi2011mean}
\begin{equation}
K_{\rm II}(r) = 2\pi\frac{\lambda_p^2}{\lambda_{\rm II}^2} \int_0^{r}k_{\rm II}(u)udu,
\end{equation}
where
\begin{equation} \label{27}
k_{\rm II}(u)\!=\!\left\{
             \begin{array}{lcl}
             {\!\!\!0} \!\!&\text{if} &\!\!\!u\!<\!\delta \\
             {\!\!\!\frac{2V_{\delta}(u)(1-e^{-\lambda_p\pi\delta^2})-2\pi\delta^2(1-e^{-\lambda_p V_{\delta}(u)})}{\lambda_p^2\pi\delta^2 V_{\delta}(u)(V_{\delta}(u)-\pi\delta^2)}} \!\!&\text{if} &\!\!\!\delta\!\leq\!u\!\leq\!2\delta  \\
             {\!\!\!\frac{\lambda_{\rm II}^2}{\lambda_p^2}} \!\!&\text{if} &\!\!\!u\!>\!2\delta,
             \end{array}
        \right.
\end{equation}
with the intensity $\lambda_{\rm II}=\frac{1-\exp(-\lambda_p\pi\delta^2)}{\pi\delta^2}$.
When $r>2\delta$ and $\lambda_{\rm II}=\lambda$, we still have $\frac{\lambda_p^2}{\lambda^2}k_{\rm II}(u)\geq 1$ as $\lambda_p\rightarrow 0$.
Therefore,
\begin{equation} \label{20}
\!\!\!\!\lim\limits_{r\rightarrow\infty}\!\!K_{\rm GPP}(r)-K_{\rm II}(r)\!\leq\!-\pi\delta^2\frac{\exp(-\pi\lambda_p\delta^2)}{1-\exp(-\pi\lambda_p\delta^2)}\!<\!0.
\end{equation}
From (\ref{17}) and (\ref{20}), it can be concluded that the MHCP is always less regular than the GPP with the same intensity as $r\rightarrow\infty$.
Figure \ref{fig:K_r} illustrates the $K$ function of the scaled $1$-GPP with $c=0.2$, in comparison with the PPP, where $K(r)=\pi r^2$, and the MHCP with type I and II for $\delta=5/4$. It can be seen that as soon as $r\gtrapprox\delta$, the GPP is a more regular point process than the other two point processes.
\ifCLASSOPTIONonecolumn
\else
\begin{figure}
    \centering
    \includegraphics[width=0.5\textwidth]{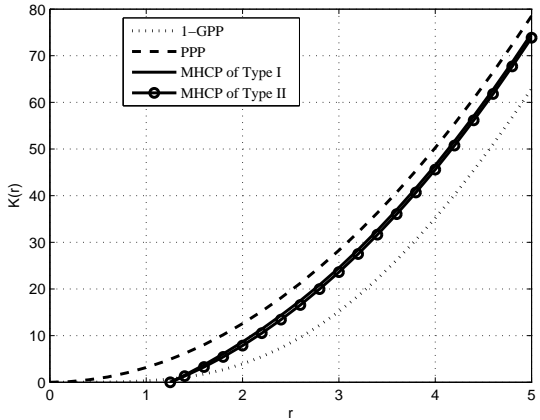}
    \caption{The $K$ function of the scaled $1$-GPP with $c=0.2$, in comparison with the Poisson point process (PPP), where $K(r)=\pi r^2$, and the MHCP with type I and II for $\delta=5/4$.}
    \label{fig:K_r}
\end{figure}
\fi
\subsubsection{L function}
\label{subsubsec:L function}
The L function is defined as $L(r)\triangleq\sqrt{\frac{K(r)}{\pi}}$, which is sometimes preferred to the $K$ function since $L(r)=r$ for the uniform PPP. For the $\beta$-GPP, $L(r)=\sqrt{r^2-\frac{\beta}{c}\Big(1-e^{-\frac{c}{\beta}r^2}\Big)}$. In order to highlight the soft-core properties of the $\beta$-GPP more clearly, we use the modified $L$ function, which is defined as $\widetilde{L}\triangleq\frac{1}{r}\sqrt{\frac{K(r)}{\pi}}-1$. For the PPP, $\widetilde{L}\equiv0$, while for the $\beta$-GPP, $\widetilde{L}(r)=\sqrt{1-\frac{\beta}{cr^2}\Big(1-e^{-\frac{c}{\beta}r^2}\Big)}-1$. Also, $\widetilde{L}(r)\rightarrow-1, \forall\beta>0$, as $r\rightarrow0$.
Figure \ref{fig:L-K_thinned} illustrates the $\widetilde{L}$ function of the $\beta$-GPP for different $\beta$. It can be seen that the $\widetilde{L}$ function of the $\beta$-GPP lies between the one of the GPP and that of the PPP, as expected.
\ifCLASSOPTIONonecolumn
\begin{figure}
  \centering
  \begin{minipage}[t]{0.5\textwidth}
    \centering
    \includegraphics[width=1\textwidth]{3_K_function_Compare.eps}
  \end{minipage}%
  \begin{minipage}[t]{0.5\textwidth}
    \centering
    \includegraphics[width=1\textwidth]{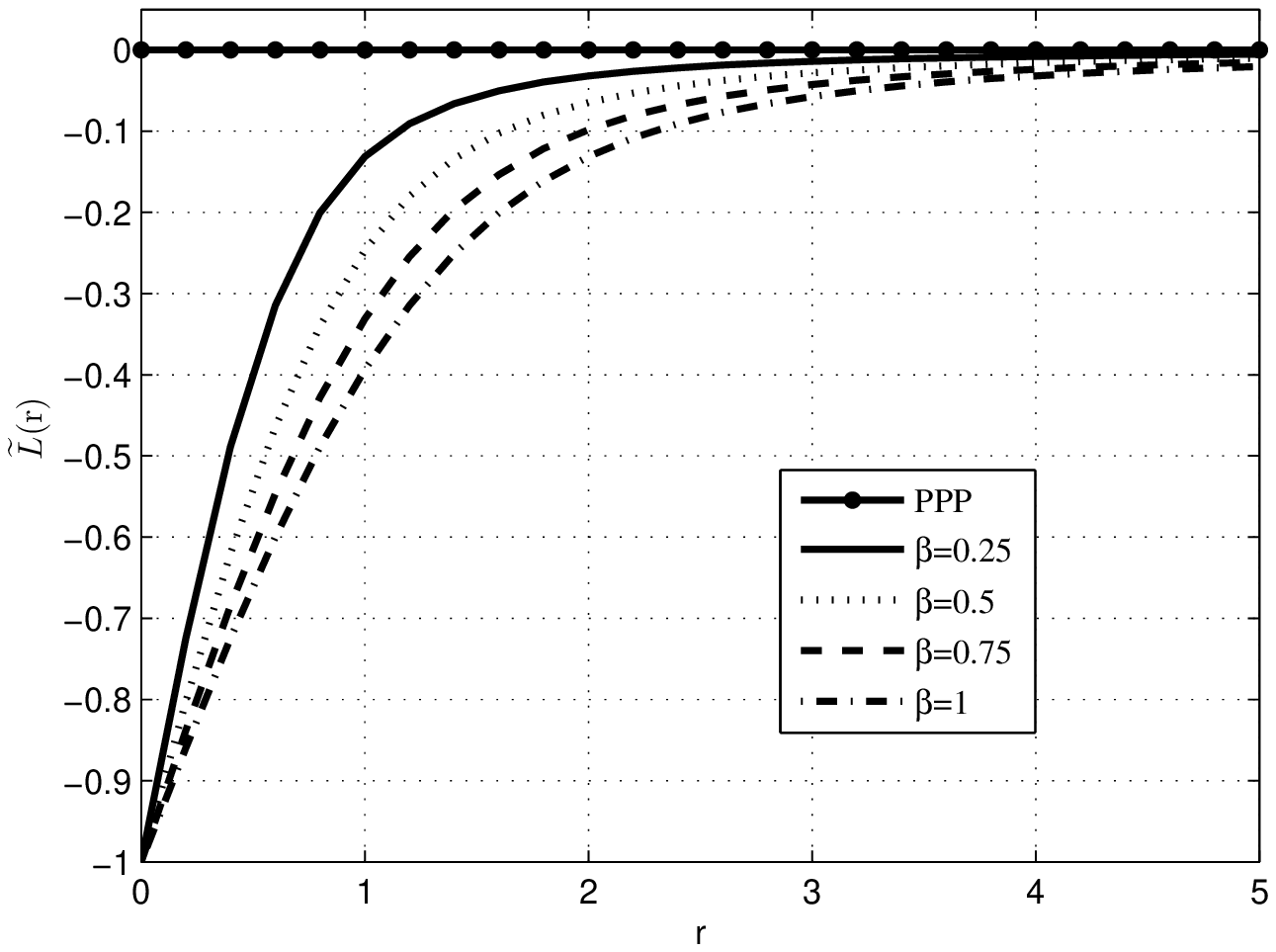}
  \end{minipage}\\[-20pt]
  \begin{minipage}[t]{0.5\textwidth}
  \caption{The $K$ function of the scaled $1$-GPP with $c=0.2$, in comparison with the Poisson point process (PPP), where $K(r)=\pi r^2$, and the MHCP with type I and II for $\delta=5/4$.}
  \label{fig:K_r}
  \end{minipage}%
  \begin{minipage}[t]{0.5\textwidth}
  \caption{The $\widetilde{L}$ function of the $\beta$-GPP.}
  \label{fig:L-K_thinned}
  \end{minipage}%
\end{figure}

\else
\begin{figure}
    \centering
    \includegraphics[width=0.5\textwidth]{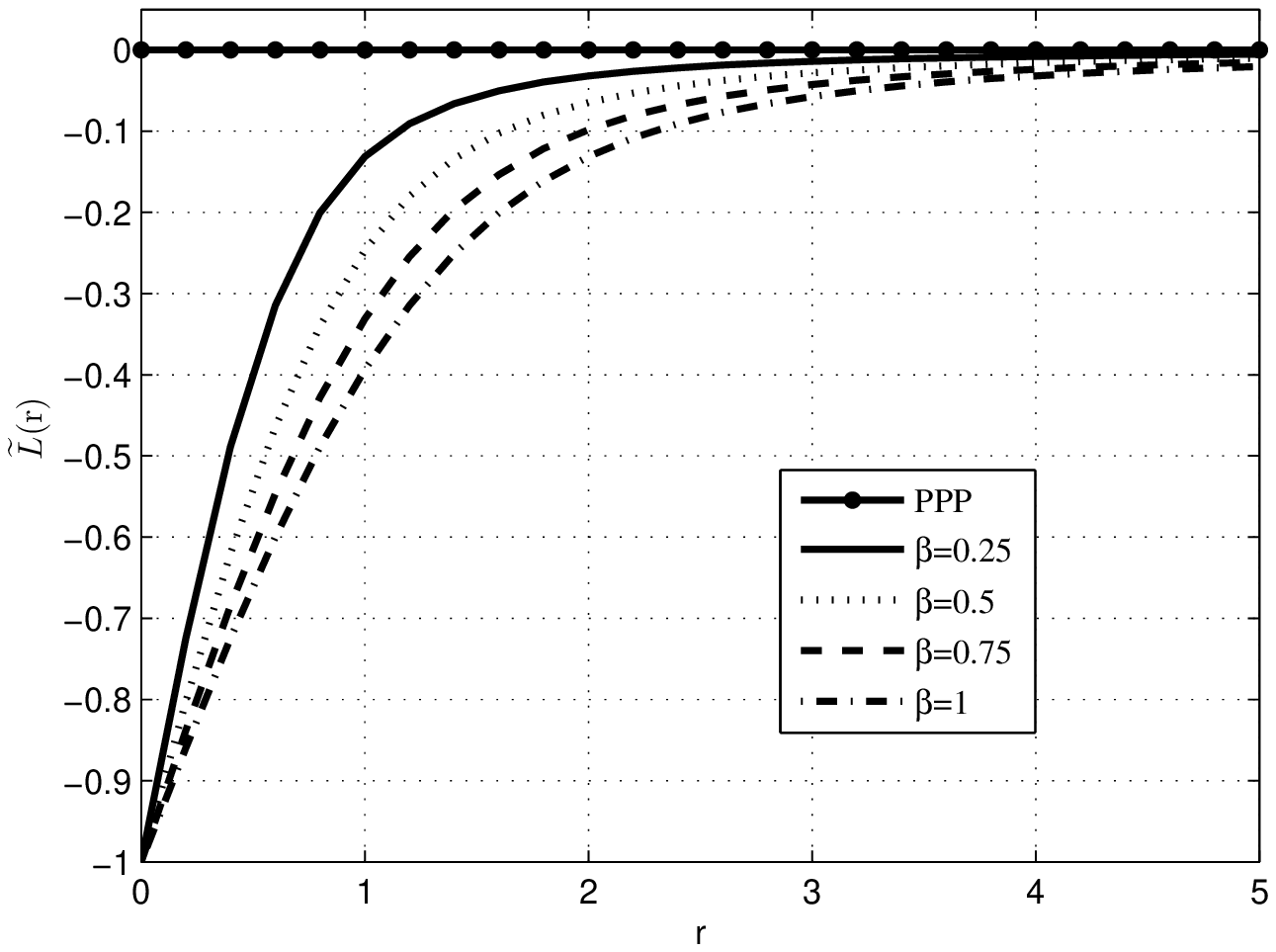}
    \caption{The $\widetilde{L}$ function of the $\beta$-GPP.}
    \label{fig:L-K_thinned}
\end{figure}
\fi

\subsection{Important distances for the $\beta$-Ginibre point process}
\label{subsec:distances}
\subsubsection{Contact distribution function or empty space function}
\label{subsubsec:contact}
The contact distance at any location $u$ of a point process $\Phi$ is $\|u-\Phi\|$ \cite[Def.~2.37]{haenggi2012stochastic}. If $\Phi_c$ is a scaled $\beta$-GPP, then the contact distribution function or empty space function $F$ is the cdf of $\|u-\Phi_c\|$:
\begin{eqnarray}
F(r)&\triangleq &\mathbb{P}(\|u-\Phi_c\|\leq r) \nonumber \\
&\overset{(a)}{=}&1-\mathbb{P}(\|o-\Phi_c\|>r) \nonumber \\
&\overset{(b)}{=}&1-\prod_{k=1}^{\infty}(\beta\mathbb{P}(Q_k\geq r^2)+1-\beta) \nonumber \\
&=&1-\prod\limits_{k=1}^{\infty}\left(1-\beta\widetilde{\gamma}\Big(k,\frac{c}{\beta}r^2\Big)\right),
\end{eqnarray}
where $(a)$ follows since the GPP is motion-invariant and thus $F(r)$ does not depend on $u$, $(b)$ is based on Proposition \ref{4}, and $\widetilde{\gamma}(a,x)=\frac{\int_0^xe^{-u}u^{a-1}du}{\Gamma(a)}$ is the normalized lower incomplete gamma function\footnote{This is the {\tt gammainc} function in Matlab.}.

\subsubsection{Nearest-neighbor distance and distribution function}
\label{subsubsec:nearest-neighbor}
The nearest-neighbor distance is the distance from a point $x\in\Phi$ to its nearest neighbor, which is given by $\|x-\Phi\setminus\{x\}\|$ \cite[Def.~2.39]{haenggi2012stochastic}. And the corresponding distribution $G^x(r)=\mathbb{P}(\|x-\Phi\setminus\{x\}\|\leq r)$ is the nearest-neighbor distance distribution function. If $\Phi_c$ is a scaled $\beta$-GPP, due to the stationarity of the GPP, we may condition the point process to have a point at the origin, i.e., $o\in\Phi_c$. According to Propositions \ref{4} and \ref{5}, we have
\begin{eqnarray}
G(r) &=& 1-\prod_{k=2}^{\infty}(\beta\mathbb{P}(Q_k\geq r^2)+1-\beta) \nonumber \\
&=&1-\prod_{k=2}^{\infty}\left(\beta\left(1-\widetilde{\gamma}\Big(k,\frac{c}{\beta}r^2\Big)\right)+1-\beta\right) \nonumber \\
&=&1-\prod_{k=2}^{\infty}\left(1-\beta\widetilde{\gamma}\Big(k,\frac{c}{\beta}r^2\Big)\right),
\end{eqnarray}
where the distribution of $Q_k$ is given in (\ref{10}).

\subsubsection{The $J$ function}
\label{subsec:J function}
The $J$ function is a useful measure of how close a process is to a PPP, which is defined as the ratio of the complementary nearest-neighbor distance and contact distributions \cite{lieshout1996nonparametric}. For the $\beta$-GPP, the $J$ function is
\begin{eqnarray} \label{21}
J(r)&\triangleq&\frac{1-G(r)}{1-F(r)}\nonumber \\
&=&\frac{1}{1-\beta+\beta e^{-\frac{c}{\beta}r^2}}.
\end{eqnarray}
It is easily verified that $J(r)\rightarrow 1$ as $\beta\rightarrow 0$, which is the $J$ function of the PPP. Figure \ref{fig:J_function_theory} gives the $J$ function of three $\beta$-GPPs. It can be seen that since the $\beta$-GPP is a soft-core point process where nodes exhibit expulsion, the $J$ function of the $\beta$-GPP is always larger than 1. Also, the increasing regularity as $\beta\rightarrow1$ is apparent from the increase in the $J$ function.

\begin{figure}
    \centering
    \includegraphics[width=0.5\textwidth]{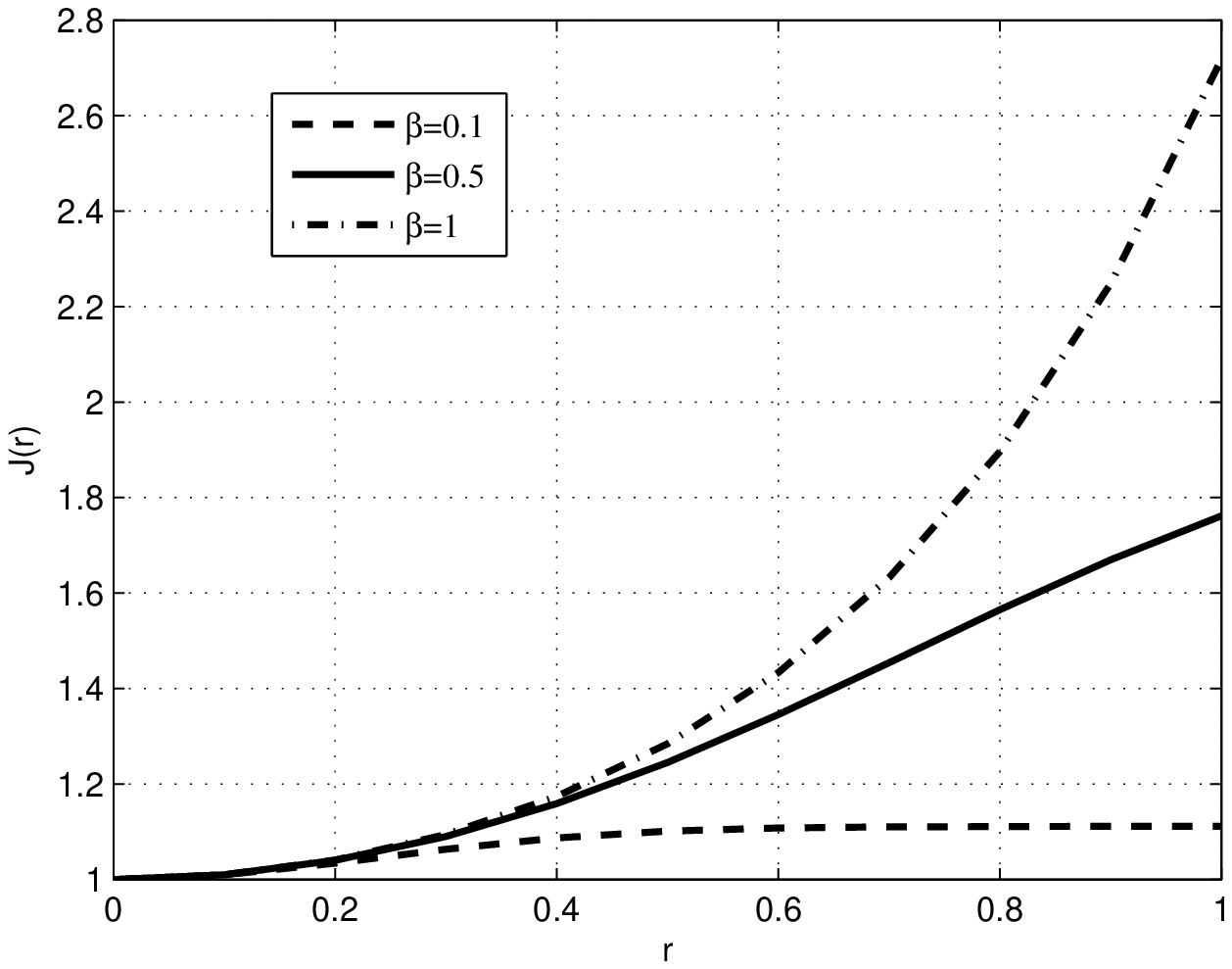}
    \caption{The J function of three $\beta$-GPPs for $c=1$.}
    \label{fig:J_function_theory}
\end{figure}

\section{Analysis of the mean and variance of the interference}
\label{sec:Analysis of the mean and variance of interference}
In this section, we provide two approaches to derive the mean and variance of the interference for the wireless networks whose nodes are distributed as a scaled $\beta$-GPP $\Phi_c=\{X_1,X_2,...\}\subset\mathbb{C}$ with intensity $\lambda=c/\pi$; one is with the aid of the Palm measure of $\Phi_c$, and the other with the reduced second moment measure of $\Phi_c$. We also provide approximations of the interference distribution using three known probability density functions (PDFs), the gamma distribution, the inverse Gaussian distribution, and the inverse gamma distribution.

Assuming all nodes in $\Phi_c$ are interfering transmitters, the interference at the origin is defined as $I\triangleq\sum_{x\in\Phi_c}h_x\ell(x)$, where $h_x$ is the power fading coefficient associated with node $x$ and $\ell(x)$ is the path loss function. It is assumed that $h_x$ follows the exponential distribution and $\mathbb{E}(h_x)=1$ for all $x\in\Phi_c$. By Campbell's theorem \cite{haenggi2012stochastic}, the mean interference $\mathbb{E}(I)$ is the same for all stationary point processes of the same intensity.
Rather than measuring interference at an arbitrary location, we focus on the interference at the location of a node $x\in\Phi_c$ but without considering that node's contribution to the interference. The statistics of the resulting interference correspond to those of the typical point at $o$ when conditioning on $o\in\Phi_c$. Consequently, the mean interference is $\mathbb{E}_o^!(I)$, where $\mathbb{E}_o^!$ is the expectation with respect to the reduced Palm distribution.

\subsection{Approach I - The Palm Measure}
From Proposition \ref{4}, we know that the set $\{|X_i|^2\}_{i\in\mathbb{N}}$ has the same distribution as the set obtained by retaining each $Q_k$ with probability $\beta$ independently of everything else. Therefore, the interference from the $i$th node is $Q_i^{-\alpha/2}T_{i}$, where $\{T_{i}\}$ is a family of independent indicators with $\mathbb{E}T_i=\beta$, $T_i\in\{0,1\}$. Besides, from Proposition \ref{5}, we know that the Palm distribution is obtained by removing $\sqrt{Q_1}$ if $Q_1$ is included in $\Xi$. The following theorem gives the mean and the variance of the interference in $\beta$-Ginibre wireless networks based on the Palm measure of $\Phi_c$.
\begin{myTheo} \label{1}
In a $\beta$-Ginibre wireless network with the bounded power-law path-loss $\ell(r)=(\max\{r_0,r\})^{-\alpha}$, the mean interference is
\ifCLASSOPTIONonecolumn
\begin{equation} \label{40}
\mathbb{E}_o^!(I)=cr_0^{2-\alpha}\frac{\alpha}{\alpha-2} + \beta r_0^{-\alpha}(e^{-\frac{c}{\beta}r_0^2}-1) -c^{\frac{\alpha}{2}}\beta^{1-\frac{\alpha}{2}}\Gamma\left(1-\frac{\alpha}{2},\frac{c}{\beta}r_0^2\right),\quad\alpha>2,
\end{equation}
\else
\begin{multline} \label{40}
\mathbb{E}_o^!(I)=cr_0^{2-\alpha}\frac{\alpha}{\alpha-2} + \beta r_0^{-\alpha}(e^{-\frac{c}{\beta}r_0^2}-1) \\ -c^{\frac{\alpha}{2}}\beta^{1-\frac{\alpha}{2}}\Gamma\left(1-\frac{\alpha}{2},\frac{c}{\beta}r_0^2\right),\quad\alpha>2,
\end{multline}
\fi
and $\mathbb{E}_o^!(I)\rightarrow a+bc$ as $c\rightarrow\infty$, for $a=-\beta r_0^{-\alpha}$, $b=\frac{\alpha}{\alpha-2}r_0^{2-\alpha}$.
The variance of the interference $V_o^!(I)\triangleq\mathbb{E}_o^!(I^2)-(\mathbb{E}_o^!(I))^2$ is given by
\ifCLASSOPTIONonecolumn
\begin{multline} \label{41}
V_o^!(I)=2\frac{c\alpha r_0^{2-2\alpha}}{\alpha-1}-2\beta r_0^{-2\alpha}(1-e^{-\frac{c}{\beta}r_0^2})-2c^{\alpha}\beta^{1-\alpha}\Gamma\Big(1-\alpha,\frac{c}{\beta}r_0^2\Big) \\
-\beta^2\sum_{k=2}^{\infty}\left(r_0^{-\alpha}\widetilde{\gamma}\Big(k,\frac{c}{\beta}r_0^2\Big)+\Big(\frac{c}{\beta}\Big)^{\frac{\alpha}{2}}\frac{\Gamma(k-\frac{\alpha}{2},\frac{c}{\beta}r_0^2)}{\Gamma(k)}\right)^2,~~~\alpha>1.
\end{multline}
\else
\begin{multline} \label{41}
\!\!\!\!\!V_o^!(I)\!=\!2\frac{c\alpha r_0^{2-2\alpha}}{\alpha-1}-2\beta \frac{1\!-\!e^{-\frac{c}{\beta}r_0^2}}{r_0^{2\alpha}}-2c^{\alpha}\beta^{1-\alpha}\Gamma\Big(1-\alpha,\frac{c}{\beta}r_0^2\Big) \\
-\beta^2\!\sum_{k=2}^{\infty}\!\left(\!r_0^{-\alpha}\widetilde{\gamma}\Big(k,\frac{c}{\beta}r_0^2\Big)\!+\!\Big(\frac{c}{\beta}\Big)^{\frac{\alpha}{2}}\!\frac{\Gamma(k-\frac{\alpha}{2},\frac{c}{\beta}r_0^2)}{\Gamma(k)}\!\right)^2\!,~\alpha\!>\!1.
\end{multline}
\fi
\end{myTheo}
The proof is provided in Appendix \ref{sect:Appendix A}. It can be shown that the variance is finite for any $r_0>0$ if $\alpha>1$. This is intuitive, since in that case, the variance is finite also for the PPP \cite[Sec.~5.1]{haenggi2012stochastic}. For two special cases, we have the following corollaries:
\begin{cor}
For $\ell(r)=r^{-\alpha}$, the mean interference is
\begin{equation}
\mathbb{E}_o^!(I)=-c^{\frac{\alpha}{2}}\beta^{1-\frac{\alpha}{2}}\Gamma\Big(1-\frac{\alpha}{2}\Big),\quad 2<\alpha<4,~\beta>0,
\end{equation}
and $\mathbb{E}_o^!(I)=\Theta(c^{\alpha/2})$, as $c\rightarrow\infty$.

The variance of the interference is upper bounded as
\begin{equation}
V_o^!(I)\leq-c^{\alpha}\beta^{1-\alpha}\!\left(2\Gamma(1-\alpha)+\beta\sum\limits_{k=2}^{\infty}k^{-\alpha}\right),~1<\alpha<2.
\end{equation}
\end{cor}
\begin{proof}
When $\ell(r)=r^{-\alpha}$, the mean interference is given by
\begin{eqnarray}
\mathbb{E}_o^!(I)\!\!\!\!&=&\!\!\!\!\sum_{k=2}^{\infty}\mathbb{E}\left(Q_k^{-\alpha/2}T_k\right)\nonumber \\
\!\!\!\!&=&\!\!\!\!\beta\sum\limits_{k=2}^{\infty}\int_0^{\infty}q^{-\alpha/2}\frac{(c/\beta)^k}{\Gamma(k)}q^{k-1}e^{-\frac{c}{\beta}q}dq \nonumber \\
\!\!\!\!&=&\!\!\!\!c\int_{0}^{\infty}q^{-\frac{\alpha}{2}}(1-e^{-\frac{c}{\beta}q})dq \nonumber \\
\!\!\!\!&\overset{(a)}{=}&\!\!\!\!c^{\frac{\alpha}{2}}\beta^{1-\frac{\alpha}{2}}\Gamma(1-\frac{\alpha}{2}),~~ 2<\alpha<4,
\end{eqnarray}
where $(a)$ follows from $\int_0^\infty x^{\nu-1}[1-\exp(-\mu x^p)]dx=-\frac{1}{|p|}\mu^{-\frac{\nu}{p}}\Gamma(\frac{\nu}{p})$, (for $\Re e(\mu)>0$ and $-p<\Re e(\nu)<0$ for $p>0$, $0<\Re e(\nu)<-p$ for $p<0$) \cite[Eq.~3.478(2)]{gradshteyn1994tables}.

The variance of the interference can be derived as
\ifCLASSOPTIONonecolumn
\begin{eqnarray}\label{49}
V_o^!(I)\!\!\!\!\!&=&\!\!\!\!\!\sum_{k=2}^{\infty}\beta\mathbb{E}(h_k^2)\mathbb{E}(Q_k^{-\alpha})-\beta^2\mathbb{E}^2\Big(Q_k^{-\alpha/2}h_k\Big) \nonumber \\
\!\!\!\!\!&=&\!\!\!\!\!2c\int_0^{\infty}q^{-\alpha}(1-e^{-\frac{c}{\beta}q})dq -\beta^2\sum_{k=2}^{\infty}\mathbb{E}^2\Big(Q_k^{-\alpha/2}\Big)\nonumber \\
\!\!\!\!\!&\overset{(b)}{\leq}&\!\!\!\!\!2c\int_0^{\infty}q^{-\alpha}(1-e^{-\frac{c}{\beta}q})dq -\beta^2\sum_{k=2}^{\infty}(\mathbb{E}(Q_k))^{-\alpha}\nonumber \\
\!\!\!\!\!&=&\!\!\!\!\!-2c^{\alpha}\beta^{1-\alpha}\Gamma(1-\alpha)
-c^{\alpha}\beta^{2-\alpha}\sum_{k=2}^{\infty}k^{-\alpha},~1<\alpha<2,
\end{eqnarray}
\else
\begin{eqnarray}\label{49}
V_o^!(I)\!\!\!\!\!&=&\!\!\!\!\!\sum_{k=2}^{\infty}\beta\mathbb{E}(h_k^2)\mathbb{E}(Q_k^{-\alpha})-\beta^2\mathbb{E}^2\Big(Q_k^{-\alpha/2}h_k\Big) \nonumber \\
\!\!\!\!\!&=&\!\!\!\!\!2c\int_0^{\infty}q^{-\alpha}(1-e^{-\frac{c}{\beta}q})dq -\beta^2\sum_{k=2}^{\infty}\mathbb{E}^2\Big(Q_k^{-\alpha/2}\Big)\nonumber \\
\!\!\!\!\!&\overset{(b)}{\leq}&\!\!\!\!\!2c\int_0^{\infty}q^{-\alpha}(1-e^{-\frac{c}{\beta}q})dq -\beta^2\sum_{k=2}^{\infty}(\mathbb{E}(Q_k))^{-\alpha}\nonumber \\
\!\!\!\!\!&=&\!\!\!\!\!-2c^{\alpha}\!\beta^{1-\alpha}\Gamma(1\!-\!\alpha)
\!-\!c^{\alpha}\!\beta^{2-\alpha}\!\!\sum_{k=2}^{\infty}\!k^{-\alpha},~1\!<\!\alpha\!<\!2, \nonumber \\
\end{eqnarray}
\fi
where $(b)$ follows from Jensen's inequality.
\end{proof}
\emph{Remarks.} From (\ref{49}), the first term is finite only when $1<\alpha<2$, and the second term is an infinite series summation with the $k$-th element $a_k=\frac{1}{k^\alpha}$, $k>1$, which is finite when $\alpha>1$ according to the finiteness of the $p$-series \cite[Eq.~0.233]{gradshteyn1994tables}. Thus, the variance of the interference is finite for $1<\alpha<2$, while the mean interference is finite for a different interval of $\alpha$, i.e., $2<\alpha<4$.

\begin{cor}
When $\alpha=4$, the mean interference is
\ifCLASSOPTIONonecolumn
\begin{equation}
\mathbb{E}_o^!(I)=2cr_0^{-2}+\beta r_0^{-4}(e^{-\frac{c}{\beta}r_0^2}-1)-\frac{c^{2}}{\beta}\mathrm{Ei}\left(-\frac{c}{\beta}r_0^2\right) - \frac{c}{r_0^2}e^{-\frac{c}{\beta}r_0^2},
\end{equation}
\else
\begin{multline}
\mathbb{E}_o^!(I)=2cr_0^{-2}+\beta r_0^{-4}(e^{-\frac{c}{\beta}r_0^2}-1) \\
-\frac{c^{2}}{\beta}\mathrm{Ei}\left(-\frac{c}{\beta}r_0^2\right) - \frac{c}{r_0^2}e^{-\frac{c}{\beta}r_0^2},
\end{multline}
\fi
and the variance of the interference is
\ifCLASSOPTIONonecolumn
\begin{equation}
V_o^!(I)=2\beta r_0^{-2}\left(e^{-\frac{c}{\beta}r_0^2}\xi+\frac{4c}{3\beta}r_0^{-4}-r_0^{-6}\right)-\frac{c^4}{3\beta^3}\mathrm{Ei}\Big(-\frac{c}{\beta}r_0^2\Big)-\beta^2\eta,
\end{equation}
\else
\begin{multline}
V_o^!(I)=2\beta r_0^{-2}\left(e^{-\frac{c}{\beta}r_0^2}\xi+\frac{4c}{3\beta}r_0^{-4}-r_0^{-6}\right) \\
-\frac{c^4}{3\beta^3}\mathrm{Ei}\Big(-\frac{c}{\beta}r_0^2\Big)-\beta^2\eta,
\end{multline}
\fi
where $\xi=r_0^{-6}-\frac{c}{3\beta}r_0^{-4}+\frac{1}{6}(\frac{c}{\beta})^2r_0^{-2}-\frac{1}{6}(\frac{c}{\beta})^3$ and $\eta=\sum\limits_{k=2}^{\infty}\left(\frac{\widetilde{\gamma}(k,\frac{c}{\beta}r_0^2)}{r_0^4}+(\frac{c}{\beta})^2\frac{\Gamma(k-2,\frac{c}{\beta}r_0^2)}{\Gamma(k)}\right)^2$.
\end{cor}
The proof mainly follows from $\int_{u}^{\infty}\frac{e^{-px}}{x^{n+1}}dx=(-1)^{n+1}\frac{p^{n}\mathrm{Ei}(-pu)}{n!}+\frac{e^{-pu}}{u^n}\sum\limits_{k=0}^{n-1}\frac{(-1)^kp^ku^k}{n(n-1)\cdots(n-k)}$, (for $p>0$) \cite[Eq.~3.351(4)]{gradshteyn1994tables}, and $\mathrm{Ei}(x)=-\int_{-x}^{\infty}t^{-1}e^{-t}dt$, $x<0$, denotes the exponential integral function.
\begin{figure*}[!t]
\centering
\includegraphics[width=\textwidth]{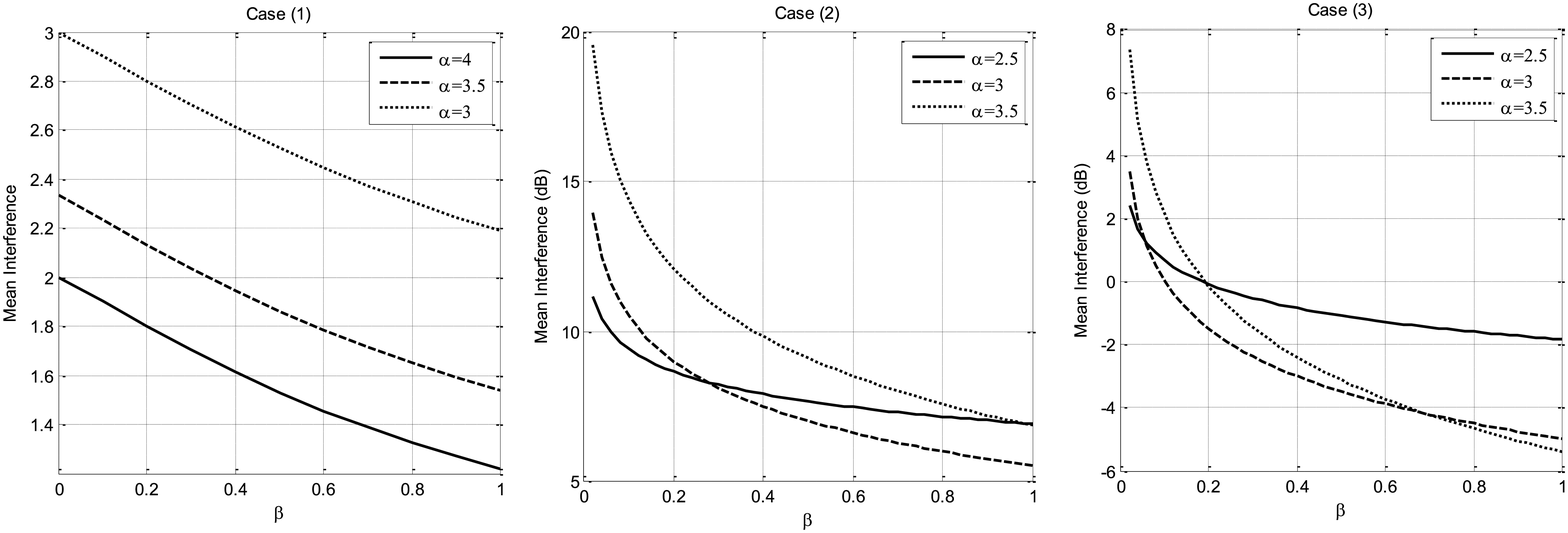}
\caption{The mean interference in $\beta$-Ginibre wireless networks for Case (1): $r_0=1$, $c=1$; Case (2): $r_0=0$, $c=1$; and Case (3): $r_0=0$, $c=0.2$.}
\label{fig:beta_mean_interference}
\end{figure*}

\ifCLASSOPTIONonecolumn
\else
\begin{figure}[!t]
    \centering
    \includegraphics[width=0.5\textwidth]{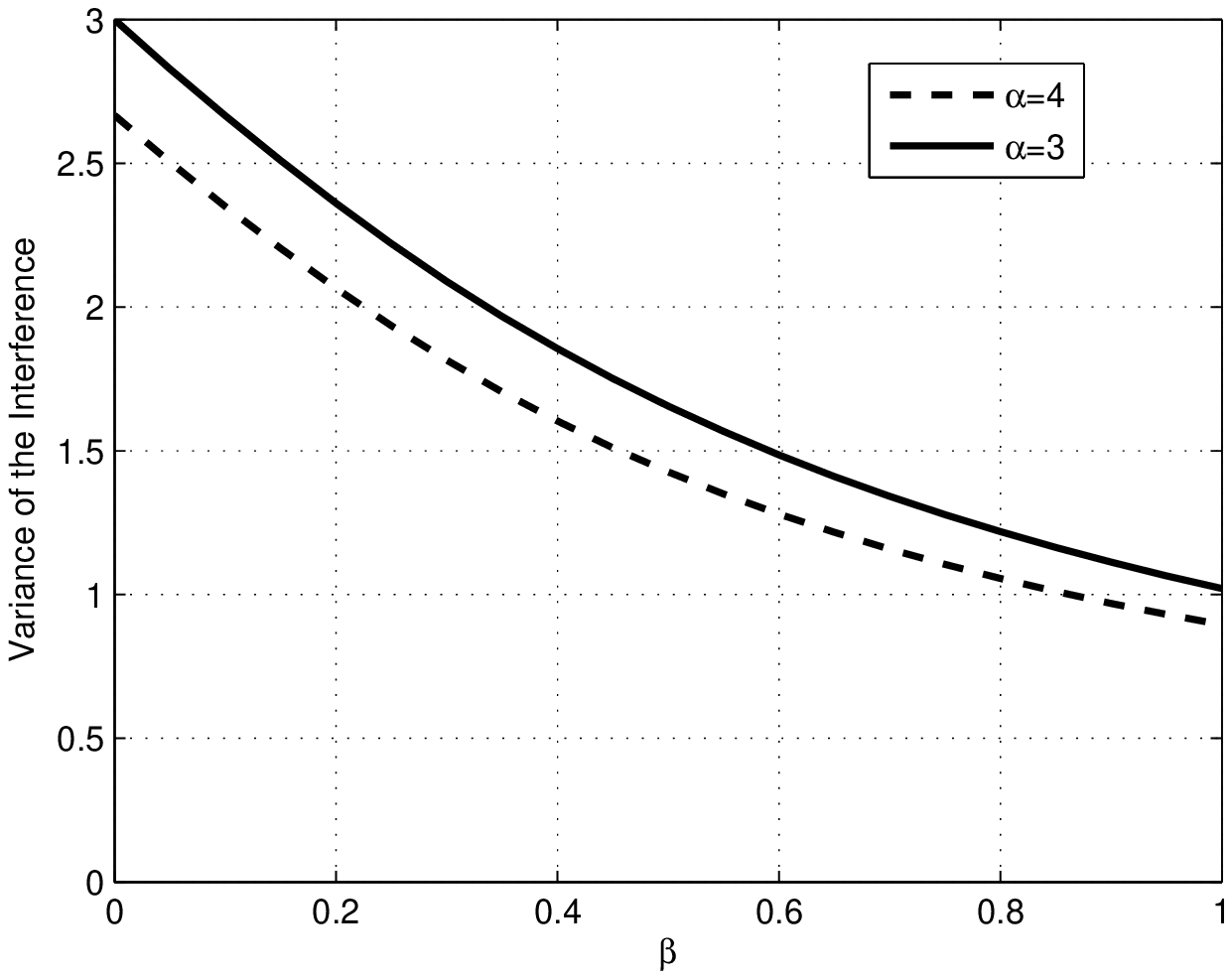}
    \caption{The variance of the interference in $\beta$-Ginibre wireless networks for $\alpha=3,4$, $c=1$ and $r_0=1$.}
    \label{fig:Variance_beta}
\end{figure}
\fi

Since the path-loss model employed is $\ell(r)=(\max\{r_0,r\})^{-\alpha}$, the interferences from nodes closer than $r_0$ are the same, irrespective of their actual distance. For nodes at distances $r>r_0$, the interference with larger $\alpha$ will attenuate faster, and that is why the mean interference decreases as $\alpha$ increases, see Figure \ref{fig:beta_mean_interference}, which shows how the mean interference of $\beta$-Ginibre wireless network changes with $\beta$ for different cases of $\alpha$, $r_0$ and $c$.

When $r_0=0$, the path-loss model $\ell(r)=r^{-\alpha}$ results in strong interference when the interfering nodes are very close to the receiver, i.e. $r\ll1$, and the
corresponding interference will increase as the path-loss exponent $\alpha$ becomes larger. Conversely, when the interfering nodes are located far from the receiver, the corresponding interference will decrease as the path-loss exponent $\alpha$ increases. From Cases (2) and (3), we can see that there are two competing effects, one is the thinning parameter $\beta$ and the other is the intensity $c$. Increasing the former one will result in the decrease of the mean interference while increasing the latter one will cause the mean interference to increase instead.

Specifically, as $\beta$ tends to zero, the
$\beta$-GPP tends to the PPP and interfering nodes are more likely to appear in the vicinity of the receiver, thus dominating the
interference power. Therefore, a larger $\alpha$ leads to a larger mean interference when $\beta$ is small. As $\beta$ increases, the repulsiveness between the
nodes will reduce the impact of the path-loss model and the interference from the far nodes decreases as $\alpha$ increases. It can be
observed from Case (3) of Figure \ref{fig:beta_mean_interference} that the curve with $\alpha=3.5$ has the largest mean interference when
$\beta$ is quite small and the smallest mean interference when $\beta>0.8$, compared with the other two curves. Meanwhile, it decreases the fastest among
the three curves. However, from Case (2), we can see that the curve with $\alpha=3.5$ is
always higher than that with $\alpha=3$, which seems to contradict our statement above. The reason is that the inter-distance
between nodes is not large enough to reduce the impact of the path-loss due to the relatively high intensity $1/\pi$ of the GPP.

Figure \ref{fig:Variance_beta} illustrates how the variance of the interference changes in $\beta$-Ginibre wireless networks with $\beta$ for $\alpha=3,4$, $c=1$ and $r_0=1$. It can be observed that the variance of the interference decreases as $\beta$ increases, which is consistent with the fact that networks with larger $\beta$ are more regular.
\ifCLASSOPTIONonecolumn
\begin{figure*}
\setcaptionwidth{0.45\textwidth}
  \centering
  \begin{minipage}[t]{0.5\textwidth}
    \centering
    \includegraphics[width=1\textwidth]{Variance_Approach2.eps}
  \end{minipage}%
  \begin{minipage}[t]{0.5\textwidth}
    \centering
    \includegraphics[width=1\textwidth]{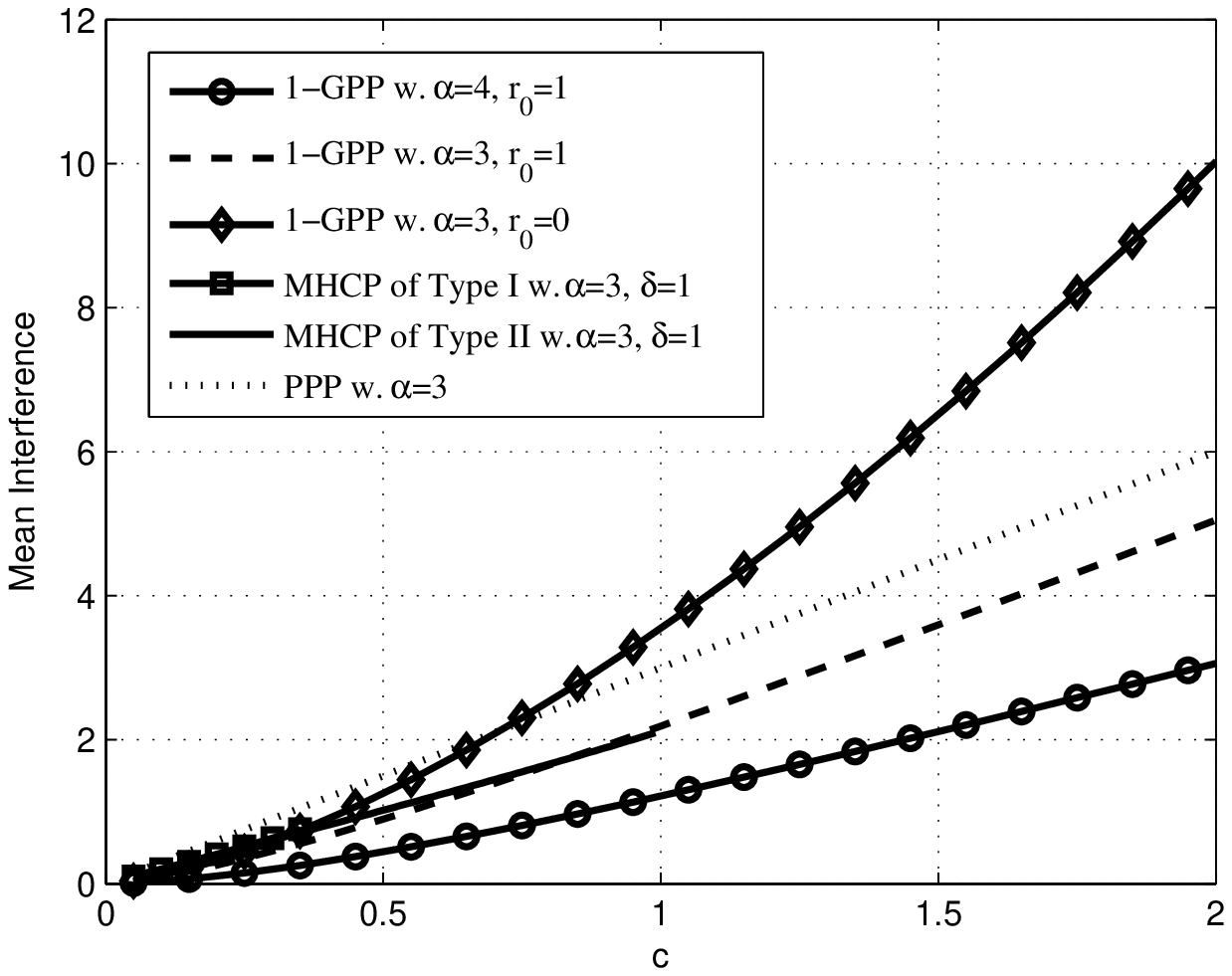}
  \end{minipage}\\[-20pt]
  \begin{minipage}[t]{0.5\textwidth}
  \caption{The variance of the interference in $\beta$-Ginibre wireless networks for $\alpha=3,4$, $c=1$ and $r_0=1$.}
  \label{fig:Variance_beta}
  \end{minipage}%
  \begin{minipage}[t]{0.5\textwidth}
  \caption{The mean interference in Ginibre wireless networks with $\beta=1$ for different $\alpha$ and $r_0$ and a comparison with the MHCP and the PPP. All processes have the same intensity $c/\pi$.}
  \label{fig:mean_interference}
  \end{minipage}%
\end{figure*}
\else
\begin{figure}
    \centering
    \includegraphics[width=0.5\textwidth]{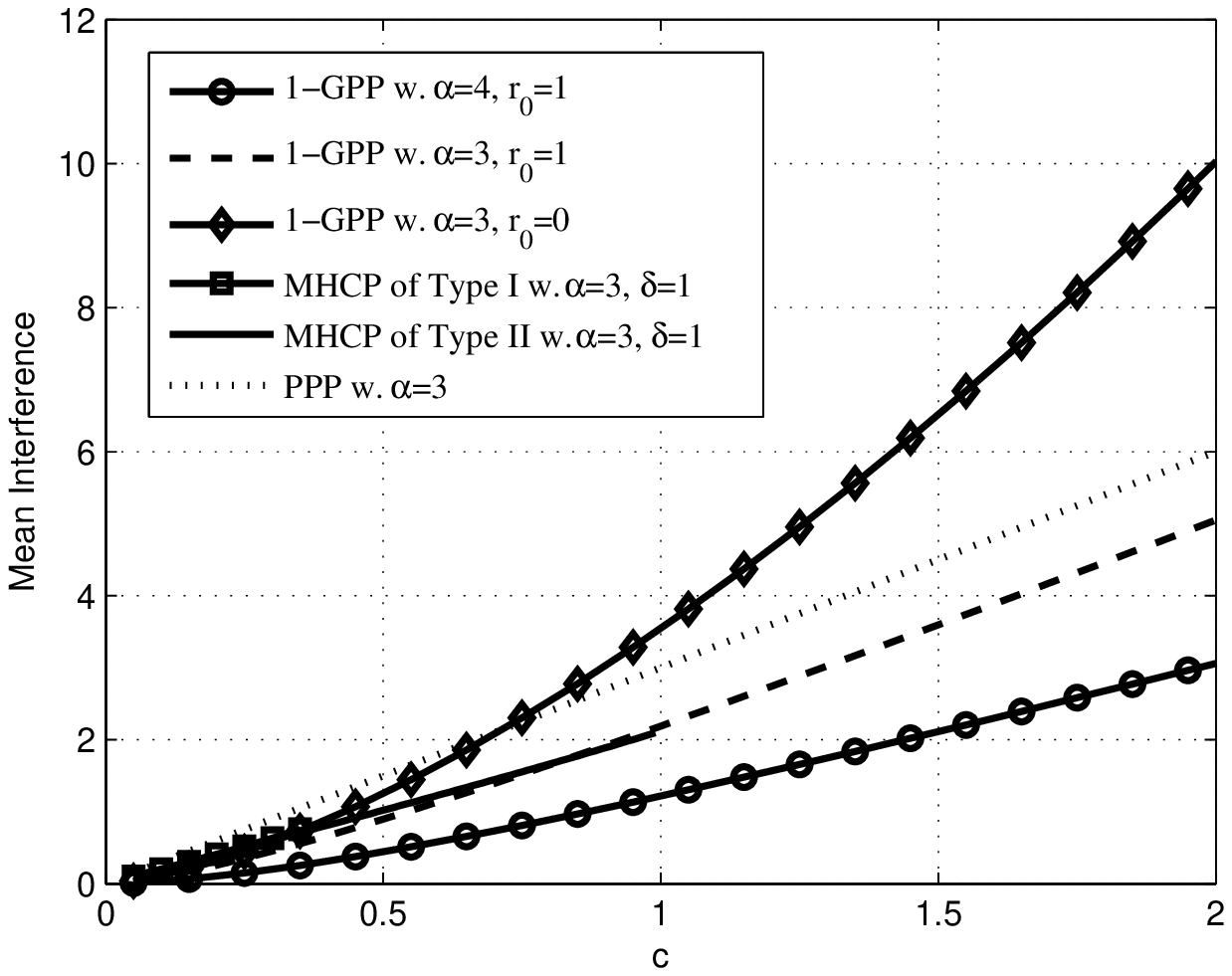}
    \caption{The mean interference in Ginibre wireless networks with $\beta=1$ for different $\alpha$ and $r_0$ and a comparison with the MHCP and the PPP. All processes have the same intensity $c/\pi$.}
    \label{fig:mean_interference}
\end{figure}
\fi
\subsection{Approach II - The Reduced Second Moment Measure}
\label{subsec:mean interference}
The aforementioned approach only applies to the GPP due to its tractable Palm measure. A more general approach to derive the mean interference is to use the reduced second moment measure, which applies to many spatial point processes. In order to further demonstrate the tractability of the GPP, we provide an alternative proof of Theorem \ref{1} using the reduced second moment measure of $\Phi_c$. Since the GPP is motion-invariant, a polar representation is convenient, and the mean interference can be expressed as \cite{haenggi2011mean}
\begin{equation} \label{29}
\mathbb{E}_{o}^{!}(I)=2\pi\int_{0}^{\infty}\ell(r)\mathcal{K}(rdr)=\lambda\int_{0}^{\infty}\ell(r)K'(r)dr.
\end{equation}
By specializing to the class of power path-loss law $\ell(r)=(\max\{r_0,r\})^{-\alpha}$, we obtain
\ifCLASSOPTIONonecolumn
\begin{eqnarray} \label{83}
\mathbb{E}_o^{!}(I)\!\!\!\!&=&\!\!\!\!\lambda\int_0^{\infty}\ell(r)K'(r)dr \nonumber \\
\!\!\!\!&=&\!\!\!\!cr_0^{2-\alpha}\frac{\alpha}{\alpha-2}+\beta r_0^{-\alpha}(e^{-\frac{c}{\beta}r_0^2}-1)-c^{\frac{\alpha}{2}}\beta^{1-\frac{\alpha}{2}}\Gamma\left(1-\frac{\alpha}{2},\frac{c}{\beta}r_0^2\right),\quad \alpha>2,
\end{eqnarray}
\else
\begin{eqnarray} \label{83}
\mathbb{E}_o^{!}(I)\!\!\!\!&=&\!\!\!\!\lambda\int_0^{\infty}\ell(r)K'(r)dr \nonumber \\
\!\!\!\!&=&\!\!\!\!cr_0^{2-\alpha}\frac{\alpha}{\alpha-2}+\beta r_0^{-\alpha}(e^{-\frac{c}{\beta}r_0^2}-1) \nonumber \\
&&-c^{\frac{\alpha}{2}}\beta^{1-\frac{\alpha}{2}}\Gamma\left(1-\frac{\alpha}{2},\frac{c}{\beta}r_0^2\right),\quad \alpha>2,
\end{eqnarray}
\fi
which is identical to the result in Theorem \ref{1}. Next, we compare the mean interference to the one of the PPP and the MHCP.

\subsubsection{Comparison with the Poisson point process}
\label{subsubsec:Mean interference of Poisson point process}
For the PPP, $K(r)=\pi r^2$, $\lambda=c/\pi$, and the mean interference is
\begin{equation}
\mathbb{E}_o^!(I)=\mathbb{E}(I)=\frac{c}{\pi}\int_0^{\infty}(\max\{r_0,r\})^{-\alpha}2\pi rdr.
\end{equation}
When $r_0\rightarrow0$, $\mathbb{E}_o^!(I)\rightarrow\infty$. \\
When $r_0>0$,
\begin{equation}
\!\!\mathbb{E}_o^!(I)\!=\!\frac{c}{\pi}\!\int_0^{r_0}\!r_0^{-\alpha}2\pi rdr\!+\!\frac{c}{\pi}\!\int_{r_0}^{\infty}\!r^{-\alpha}2\pi rdr\!=\!\frac{\alpha r_0^{2-\alpha}}{\alpha-2}c.
\end{equation}
Thus, for the PPP, $\mathbb{E}_o^!(I)=\Theta(c)$, as (of course) for the $\beta$-GPP when $\beta\rightarrow0$. And the difference between the mean interferences of the $1$-GPP and the PPP can be expressed as
\ifCLASSOPTIONonecolumn
\begin{eqnarray}
\triangle\mathbb{E}_o^!(I)\!\!\!\!&=&\!\!\!\!\frac{c}{\pi}\int_0^{\infty}(\max\{r_0,r\})^{-\alpha}2\pi rdr-\frac{c}{\pi}\int_0^{\infty}(\max\{r_0,r\})^{-\alpha}2\pi r(1-e^{-cr^2})dr\nonumber \\
\!\!\!\!&=&\!\!\!\!r_0^{-\alpha}(1-e^{-cr_0^2})+c^{\frac{\alpha}{2}}\int_{cr_0^2}^{\infty}r^{-\frac{\alpha}{2}}e^{-r}dr.
\end{eqnarray}
\else
\begin{eqnarray}
\!\!\triangle\mathbb{E}_o^!(I)\!\!\!\!\!&=&\!\!\!\!\!\frac{c}{\pi}\!\int_0^{\infty}\!(\max\{r_0,r\})^{-\alpha}2\pi rdr \nonumber \\
&&-\frac{c}{\pi}\!\int_0^{\infty}\!(\max\{r_0,r\})^{-\alpha}2\pi r(1-e^{-cr^2})dr \nonumber \\
\!\!\!\!\!&=&\!\!\!\!\!r_0^{-\alpha}(1-e^{-cr_0^2})+c^{\frac{\alpha}{2}}\int_{cr_0^2}^{\infty}r^{-\frac{\alpha}{2}}e^{-r}dr.
\end{eqnarray}
\fi
Letting $f(c)=\triangle\mathbb{E}_o^!(I)$, the first derivative of $f(c)$ can be obtained as
\begin{equation}
f'(c)=\frac{\alpha}{2}c^{\frac{\alpha}{2}-1}\int_{cr_0^2}^{\infty}r^{-\frac{\alpha}{2}}e^{-r}dr.
\end{equation}
Therefore, $f(c)$ is an increasing function since $f'(c)>0$, and when $c\rightarrow\infty$, the maximum of $\triangle\mathbb{E}_o^!(I)$ is obtained as $r_0^{-\alpha}$, i.e., $r_0^{-\alpha}$ is the largest gap between the interferences of the 1-GPP and the PPP.

\subsubsection{Comparison with the Mat\'ern hard-core process}
\label{subsubsec:Mean interference of the hard-core process}
For the MHCP of type I, the mean interference is
\begin{equation}
\mathbb{E}_{o}^{!}(I)=2\pi\lambda_p\exp(\lambda_p\pi\delta^2)\int_{\delta}^{\infty}\ell(r)k_{\rm I}(r)rdr,
\end{equation}
while for the MHCP of type II, the mean interference is
\begin{equation}
\mathbb{E}_{o}^{!}(I)=\frac{2\pi}{\lambda_{\rm II}} \int_{\delta}^{\infty}\ell(r)\lambda_p^2k_{\rm II}(r)rdr.
\end{equation}
When $\lambda_p\rightarrow \infty$, $\lambda_{\rm II}\rightarrow \frac{1}{\pi\delta^2}$, the mean interference is finite \cite{haenggi2011mean}.

Figure \ref{fig:mean_interference} illustrates the mean interference in Ginibre wireless networks with $\beta=1$ for different $\alpha$ and $r_0$, in comparison with the MHCP of type I and II ($\delta=1$) and the PPP. It can be seen that except for the PPP, the mean interferences of the point processes are not proportional to $c$, which means the mean interferences of soft-core and hard-core processes are not proportional to their intensities. From (\ref{29}), whether the mean interference is proportional to $c$ depends on whether the reduced second measure (i.e., the expected number of the interfering sources within distance $r$ of the origin) is proportional to $c$. For the GPP and the MHCP, their reduced second measures are not proportional to the intensity; while for the PPP, where $\mathcal{K}(b(o,r))=cr^2$, its mean interference is proportional to $c$, as shown in the figure.

In the following, we provide an alternative derivation of the variance of the interference for the $\beta$-GPPs. We have
$V_o^!(I)=\mathbb{E}_o^{!}(I^2)-(\mathbb{E}_o^{!}(I))^2$, where
\ifCLASSOPTIONonecolumn
\begin{eqnarray} \label{16}
\mathbb{E}_o^{!}(I^2)\!\!\!\!&=&\!\!\!\!\mathbb{E}_o^{!}\left(\sum_{x\in\Phi}\ell(x)h_x\sum_{y\in\Phi}\ell(y)h_y\right)\nonumber\\
\!\!\!\!&=&\!\!\!\!\mathbb{E}_o^{!}\left(\sum_{x\in\Phi}\ell(x)^2h_x^2\right)+\mathbb{E}_o^{!}\left(\sum_{x,y\in\Phi}^{\neq}\ell(x)\ell(y)h_xh_y\right)\nonumber\\
\!\!\!\!&\overset{(a)}{=}&\!\!\!\!\mathbb{E}(h_x^2)\int_{\mathbb{R}^2}\ell^2(x)\mathcal{K}(dx)+\frac{1}{\lambda}\int_{\mathbb{R}^2}\int_{\mathbb{R}^2}\ell(x)\ell(y)\varrho_{\beta,c}^{(3)}(o,x,y)dxdy,
\end{eqnarray}
\else
\begin{eqnarray} \label{16}
\mathbb{E}_o^{!}(I^2)\!\!\!\!\!\!&=&\!\!\!\!\!\!\mathbb{E}_o^{!}\left(\sum_{x\in\Phi}\ell(x)h_x\sum_{y\in\Phi}\ell(y)h_y\right)\nonumber\\
\!\!\!\!\!\!&=&\!\!\!\!\!\!\mathbb{E}_o^{!}\left(\sum_{x\in\Phi}\ell(x)^2h_x^2\right)+\mathbb{E}_o^{!}\left(\sum_{x,y\in\Phi}^{\neq}\ell(x)\ell(y)h_xh_y\right)\nonumber\\
\!\!\!\!\!\!&\overset{(a)}{=}&\!\!\!\!\!\!\mathbb{E}(h_x^2)\!\!\!\int_{\mathbb{R}^2}\!\!\!\!\ell^2\!(\!x\!)\mathcal{K}(\!dx\!)\!\!+\!\!\frac{1}{\lambda}\!\!\int_{\mathbb{R}^2}\!\!\int_{\mathbb{R}^2}\!\!\!\!\ell(\!x\!)\ell(y)\varrho_{\beta,c}^{(3)}\!(\!o,\!x,\!y\!)dxdy, \nonumber \\
\end{eqnarray}
\fi
and $(a)$ follows from the Lemma 5.2 in \cite{haenggi2009interference}.
For the first term,
\ifCLASSOPTIONonecolumn
\begin{eqnarray} \label{88}
\mathbb{E}(h_x^2)\int_{\mathbb{R}^2}\ell^2(x)\mathcal{K}(dx)\!\!\!\!&=&\!\!\!\!2\lambda\int_0^{\infty}\ell^2(r)K'(r)dr \nonumber \\
\!\!\!\!&=&\!\!\!\!2r_0^{-2\alpha}\left(\frac{\alpha cr_0^2}{\alpha-1}+\beta e^{-\frac{c}{\beta}r_0^2}-\beta\right)-2\frac{c^{\alpha}}{\beta^{\alpha-1}}\Gamma\Big(1-\alpha,\frac{c}{\beta}r_0^2\Big).
\end{eqnarray}
\else
\begin{eqnarray} \label{88}
\mathbb{E}(h_x^2)\int_{\mathbb{R}^2}\ell^2(x)\mathcal{K}(dx)\!\!\!\!&=&\!\!\!\!2\lambda\int_0^{\infty}\ell^2(r)K'(r)dr \nonumber \\
\!\!\!\!&=&\!\!\!\!2r_0^{-2\alpha}\left(\frac{\alpha cr_0^2}{\alpha-1}+\beta e^{-\frac{c}{\beta}r_0^2}-\beta\right) \nonumber \\
&&-2\frac{c^{\alpha}}{\beta^{\alpha-1}}\Gamma\Big(1-\alpha,\frac{c}{\beta}r_0^2\Big).
\end{eqnarray}
\fi
For the second term,
\ifCLASSOPTIONonecolumn
\begin{equation} \label{89}
\frac{1}{\lambda}\int_{\mathbb{R}^2}\int_{\mathbb{R}^2}\ell(x)\ell(y)\varrho_{\beta,c}^{(3)}(o,x,y)dxdy=\frac{2\pi}{c}\int_0^{2\pi}g(\theta)(2\pi-\theta)d\theta,
\end{equation}
\else
\begin{multline} \label{89}
\frac{1}{\lambda}\int_{\mathbb{R}^2}\int_{\mathbb{R}^2}\ell(x)\ell(y)\varrho_{\beta,c}^{(3)}(o,x,y)dxdy \\
=\frac{2\pi}{c}\int_0^{2\pi}g(\theta)(2\pi-\theta)d\theta,
\end{multline}
\fi
where
\ifCLASSOPTIONonecolumn
\begin{equation} \label{19}
g(\theta)\triangleq g(\theta_1,\theta_2)\!\!=\!\!\int_{0}^{\infty}\!\!\int_{0}^{\infty}\!(\max\{r_0,r_1\})^{-\alpha}(\max\{r_0,r_2\})^{-\alpha}\varrho_{\beta,c}^{(3)}(o,r_1e^{j\theta_1},r_2e^{j\theta_2})r_1r_2dr_1dr_2,
\end{equation}
\else
\begin{eqnarray} \label{19}
\!\!g(\theta)\!\!\!\!\!&\triangleq &\!\!\!\!\!g(\theta_1,\theta_2)\!=\!\!\!\int_{0}^{\infty}\!\!\!\!\int_{0}^{\infty}\!\!(\max\{r_0,r_1\})^{-\alpha}(\max\{r_0,r_2\})^{-\alpha} \nonumber \\
&&~~~~~~~~~~~\times\varrho_{\beta,c}^{(3)}(o,r_1e^{j\theta_1},r_2e^{j\theta_2})r_1r_2dr_1dr_2,
\end{eqnarray}
\fi
since $g(\theta_1,\theta_2)$ is merely related to $\theta_1-\theta_2$ from (\ref{19}), thus we let $\theta=\theta_1-\theta_2$.
By substituting (\ref{83}) and (\ref{16}) into $V_o^!(I)$, we can obtain the variance of the interference, which is identical to the one in Theorem \ref{1}.

\begin{figure}
  \centering
  \subfigure[$\alpha=3$.]{
    \label{fig:Approximation:a}
    \includegraphics[width=0.45\textwidth]{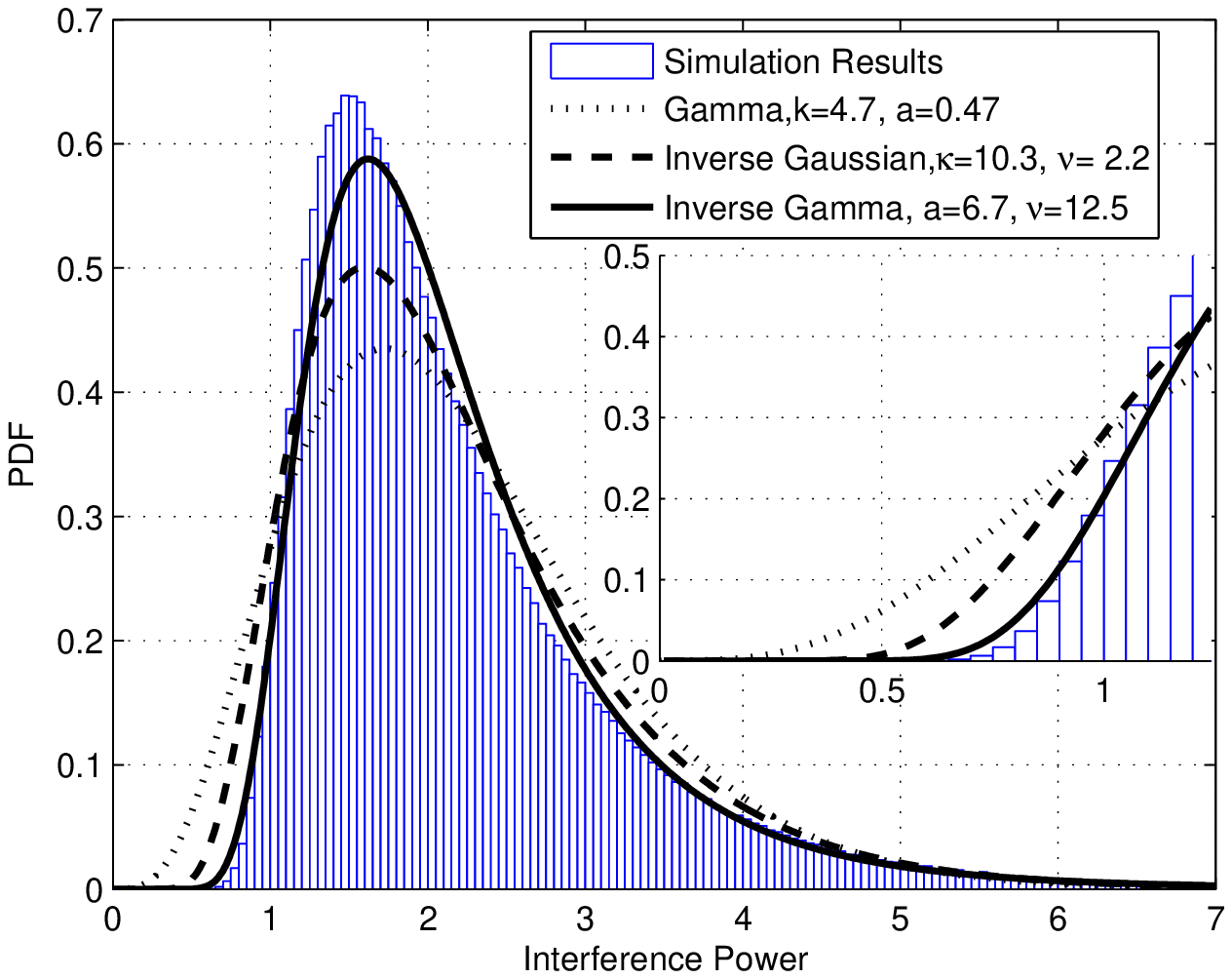}}
  \subfigure[$\alpha=4$.]{
    \label{fig:Approximation:b}
    \includegraphics[width=0.45\textwidth]{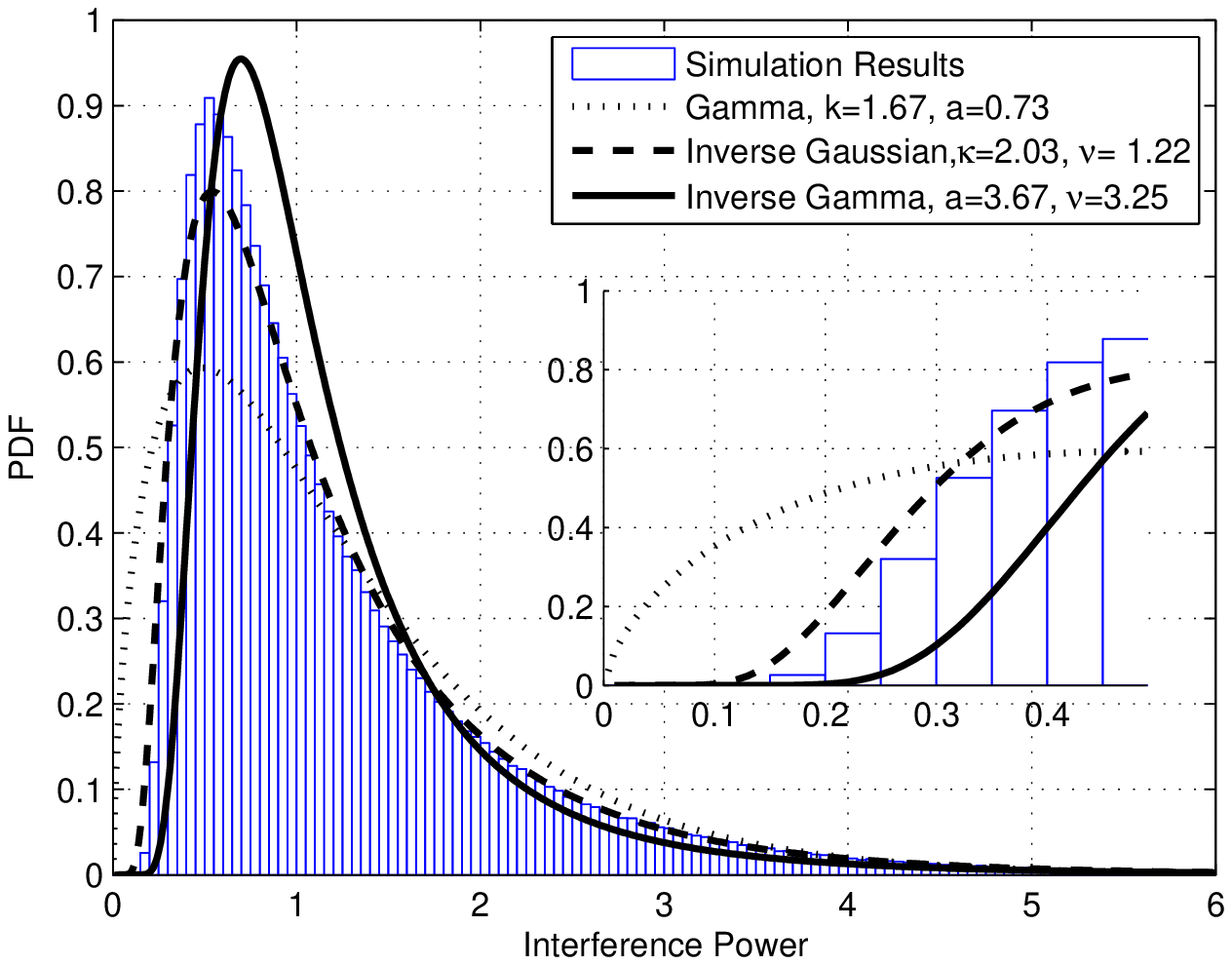}}
  \caption{Empirical PDF of the interference and the corresponding fits.}
  \label{fig:Approximation} 
\end{figure}

\subsection{Approximation of the Interference Distribution}
Based on the mean and variance of the interference in Theorem \ref{1}, we choose known probability density functions (PDFs) to approximate the PDF of the interference for $c=1$, $\beta=1$ and $r_0=1$. The three candidate densities are the gamma distribution, the inverse Gaussian distribution, and the inverse gamma distribution, given as follows \cite[Chap.~5]{haenggi2009interference},

(1) Gamma distribution:$\\$
$f(x)=x^{k-1}\exp(-x/a)/(a^k\Gamma(k))$ with mean $ka$ and variance $ka^2$;

(2) Inverse Gaussian distribution:$\\$
$f(x)=\left(\frac{\kappa}{2\pi x^3}\right)^{1/2}\exp\left(-\frac{\kappa(x-\nu)^2}{2\nu^2x}\right)$ with mean $\nu$ and variance $\nu^3/\kappa$;

(3) Inverse gamma distribution: $\\$
$f(x)=\nu^ax^{-a-1}\exp(-\nu/x)/\Gamma(a)$ with mean $\nu/(a-1)$ and variance $\nu^2/((a-1)^2(a-2))$.

From the mean and variance derived for the $\beta$-Ginibre wireless networks, we can determine the parameters of the three distributions. In Figure \ref{fig:Approximation}, we have plotted the PDFs of the interference using Monte-Carlo simulation when the underlying node distribution is $\beta$-GPP and the fading is Rayleigh with the bounded path-loss exponent $\alpha=3,4$.
We can observe that the empirical PDF around the origin is 0 in both cases, so fitting a PDF whose first derivative at the origin is 0 may give a good fit to the interference distribution for small interference situations. The derivatives of both the inverse Gaussian distribution and the inverse gamma distribution at the origin are 0 for $\alpha=3,4$ cases. For the gamma distribution, its first derivative at the origin is 0 when $\alpha=3$ but tends to $\infty$ when $\alpha=4$. Therefore, the gamma distribution is less suitable since it is not flat enough at 0 in the case of $\alpha=4$.

According to \cite{haenggi2009interference}, the interference distribution of any motion-invariant point process has exponential decay on the origin and an exponential tail due to the Rayleigh fading. However, the gamma distribution has a $k\!-\!1$th order of decay at the origin and an exponential tail, which leads to a bad fit for $\alpha=3$. Since the inverse Gaussian distribution has an exponential decay at the origin and a slightly super-exponential tail, it gives a good fit in both cases, especially for $\alpha=4$.
The inverse gamma distribution has an exponential decay at the origin and a $a\!+\!1$th order decay at the tail.
When $\alpha=3$, the inverse gamma distribution has a 7.7th order decay at the tail, which approximates the exponential tail well
in the range $[1,7]$ due to the relatively large order, thus it provides the best fit; while for $\alpha=4$, the inverse gamma distribution has a 4.67th order decay of its tail which has a big deviation from the exponential tail, thus providing a bad fit.

\section{Coverage probability}
\label{sec:Coverage probability}
In this section, we derive the coverage probability of the typical user in $\beta$-Ginibre wireless cellular networks where each user is associated with the closest base station. Due to the motion-invariance of the GPPs, we can take the origin as the location of the typical user. The power fading coefficient associated with node $X_i$ is denoted by $h_i$, which is an exponential random variable with mean $1$ (Rayleigh fading), and $h_i$, $i\in\mathbb{N}$, are mutually independent and also independent of the scaled $\beta$-GPP $\Phi_c$. The path-loss function is given by $\ell(r)=r^{-\alpha}$, for $\alpha>2$. In the setting described above, the received signal-to-interference-plus-noise ratio (SINR) of the typical user is
\begin{eqnarray}
\mathrm{SINR}\!\!\!\!&=&\!\!\!\!\frac{h_o\ell(|X_{B_o}|)}{\sigma^2+\sum\limits_{i\in\mathbb{N}\setminus{B_o}}h_i\ell(|X_i|)} \nonumber \\
\!\!\!\!&\overset{(a)}{=}&\!\!\!\!\frac{h_o\ell(|X_{B_o}|)}{\sigma^2+\sum\limits_{k\in\mathbb{N}\setminus{B_o}}h_k Q_k^{-\alpha/2}T_k},
\end{eqnarray}
where $(a)$ follows from Proposition \ref{4}, $B_o$ denotes the index of the base station associated with the typical user, and $\sigma^2$ denotes the thermal noise at the origin. The following theorem provides the coverage probability $\mathbb{P}(\mathrm{SINR}>\theta)$ of the typical user (or, equivalently, the covered area fraction at SINR threshold $\theta$).
\begin{myTheo} \label{cov}
For an SINR threshold $\theta$, the coverage probability of the typical user in the $\beta$-Ginibre wireless network is given by
\ifCLASSOPTIONonecolumn
\begin{equation} \label{23}
p(\theta,\alpha,\beta)=\beta\!\!\int_0^{\infty}\!e^{-s}e^{-\mu\theta\sigma^2(\frac{\beta s}{c})^{\alpha/2}}M(\theta,s,\alpha,\beta)S(\theta,s,\alpha,\beta)ds,
\end{equation}
\else
\begin{multline} \label{23}
p(\theta,\alpha,\beta)= \\
\beta\!\!\int_0^{\infty}\!e^{-s}e^{-\mu\theta\sigma^2(\frac{\beta s}{c})^{\alpha/2}}M(\theta,s,\alpha,\beta)S(\theta,s,\alpha,\beta)ds,
\end{multline}
\fi
where
\ifCLASSOPTIONonecolumn
\begin{equation} \label{24}
M(\theta,s,\alpha,\beta)=\prod\limits_{k=1}^{\infty}\int_s^{\infty}\frac{v^{k-1}e^{-v}}{(k-1)!}\left( \frac{\beta}{1\!+\!\theta (s/v)^{\alpha/2}}\!+\!1\!-\!\beta\right)dv,
\end{equation}
\else
\begin{multline} \label{24}
M(\theta,s,\alpha,\beta)= \\
\prod\limits_{k=1}^{\infty}\int_s^{\infty}\frac{v^{k-1}e^{-v}}{(k-1)!}\left( \frac{\beta}{1\!+\!\theta (s/v)^{\alpha/2}}\!+\!1\!-\!\beta\right)dv,
\end{multline}
\fi
\ifCLASSOPTIONonecolumn
\begin{equation} \label{25}
S(\theta,s,\alpha,\beta)=\sum\limits_{i=1}^{\infty}s^{i-1}\left(\int_s^{\infty}v^{i-1}e^{-v}\left( \frac{\beta}{1\!+\!\theta (s/v)^{\alpha/2}}\!+\!1\!-\!\beta\right)dv\right)^{-1}.
\end{equation}
\else
\begin{multline} \label{25}
S(\theta,s,\alpha,\beta)= \\
\sum\limits_{i=1}^{\infty}s^{i-1}\left(\int_s^{\infty}v^{i-1}e^{-v}\left( \frac{\beta}{1\!+\!\theta (s/v)^{\alpha/2}}\!+\!1\!-\!\beta\right)dv\right)^{-1}.
\end{multline}
\fi
\end{myTheo}
The proof is provided in Appendix \ref{sect:Appendix C}. For $\beta=1$, we retrieve the result in \cite{miyoshi2012gpp}.

\section{Fitting the $\beta$-GPP to Actual Network Deployments}
\label{sec:Fitting the beta-GPP to actual network deployments}
Modeling the spatial structure of cellular networks is one of the most important applications of the $\beta$-GPPs. In this section, we fit the $\beta$-GPP to the locations of the BSs in real cellular networks, shown in Figure \ref{fig:Urban_region} and \ref{fig:Rural_region}, obtained from the Ofcom\footnote{Ofcom - the independent regulator and competition authority for the UK communications industries, where the data are open to the public. Website: http://sitefinder.ofcom.org.uk/search} by carefully adjusting the parameter $\beta$ which controls the repulsion between the points.
Table \ref{table 1} gives the details of the two point sets. The density of each data set is estimated through calculating the total number of points divided by the area of each region. The intensity of the $\beta$-GPP is then set to the density of the given point sets. Based on the estimated density, we fit the $\beta$-GPP to the actual BS deployments with three different metrics, the $L$ function, the $J$ function and the coverage probability. The first two are classical statistics in stochastic geometry, while the last one is a key performance metric of cellular networks. Our objective of fitting is to minimize the vertical average squared error between the metrics obtained from the experimental data and those of the $\beta$-GPP. The vertical average squared error is given by
\begin{equation}
E(a,b)=\int_{a}^{b}\Big(M_{e}(t)-M_{\beta}(t)\Big)^2dt,
\end{equation}
where $a,b\in\mathbb{R}$, $M_{e}(t)$ is the curve of the experimental data and $M_{\beta}(t)$ is the curve corresponding to the $\beta$-GPP.
\ifCLASSOPTIONonecolumn
\begin{figure*}
  \centering
  \begin{minipage}[t]{0.5\textwidth}
    \centering
    \includegraphics[width=1\textwidth]{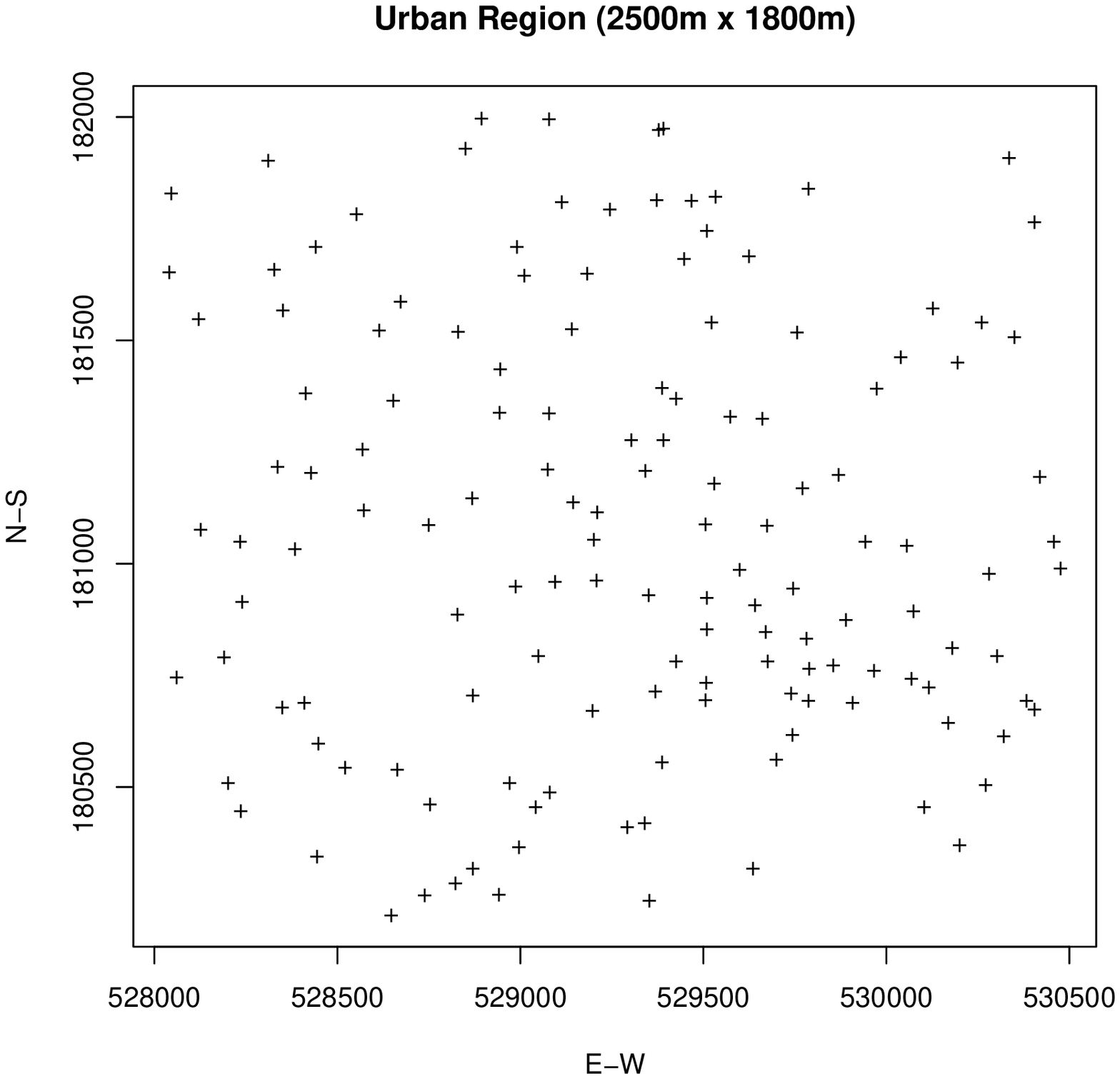}
  \end{minipage}%
  \begin{minipage}[t]{0.5\textwidth}
    \centering
    \includegraphics[width=1\textwidth]{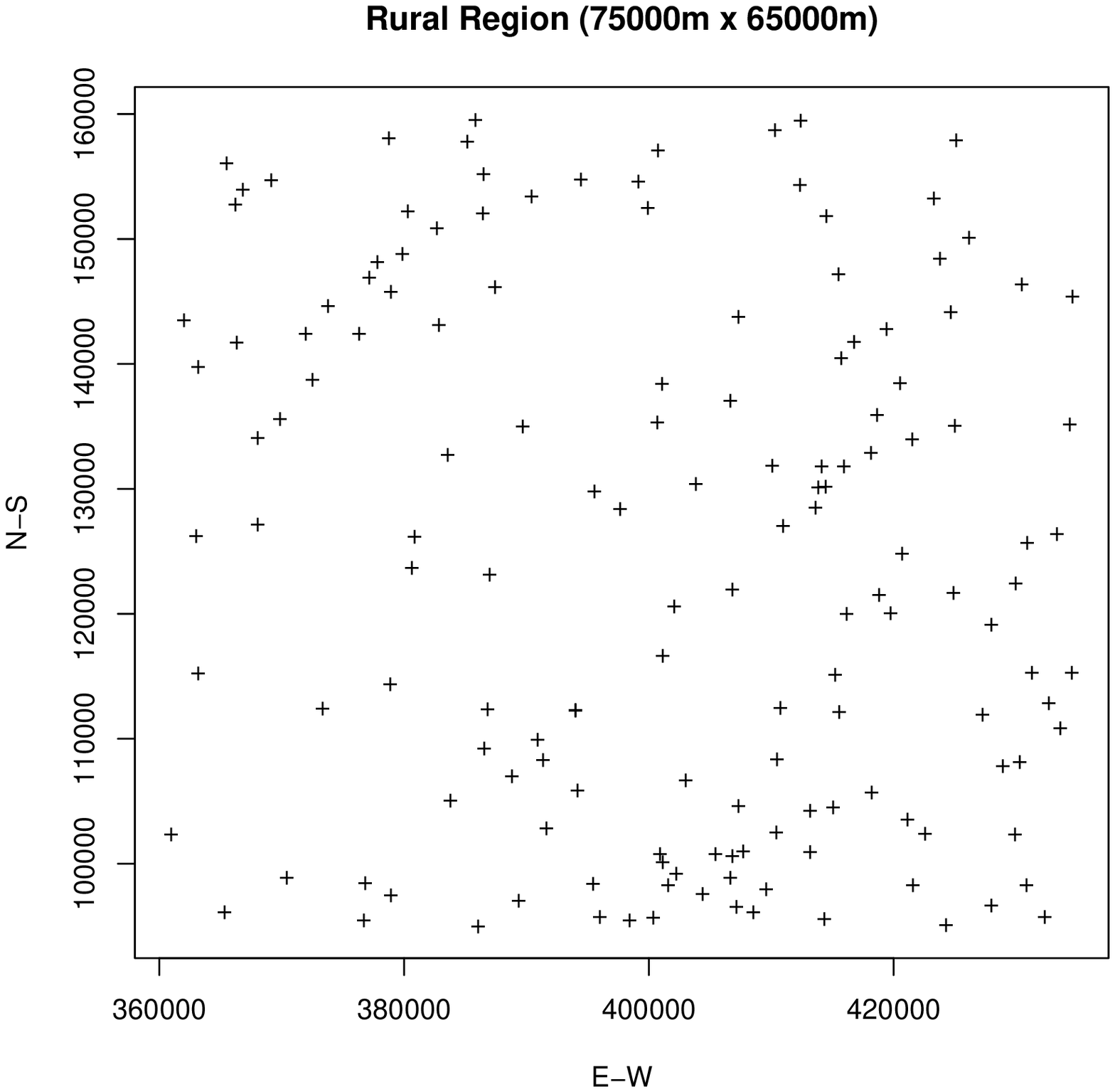}
  \end{minipage}\\[-20pt]
  \begin{minipage}[t]{0.5\textwidth}
  \caption{The locations of BSs in the urban region.}
  \label{fig:Urban_region}
  \end{minipage}%
  \begin{minipage}[t]{0.5\textwidth}
  \caption{The locations of BSs in the rural region.}
  \label{fig:Rural_region}
  \end{minipage}%
\end{figure*}

\else
\begin{figure}
    \centering
    \includegraphics[width=0.5\textwidth]{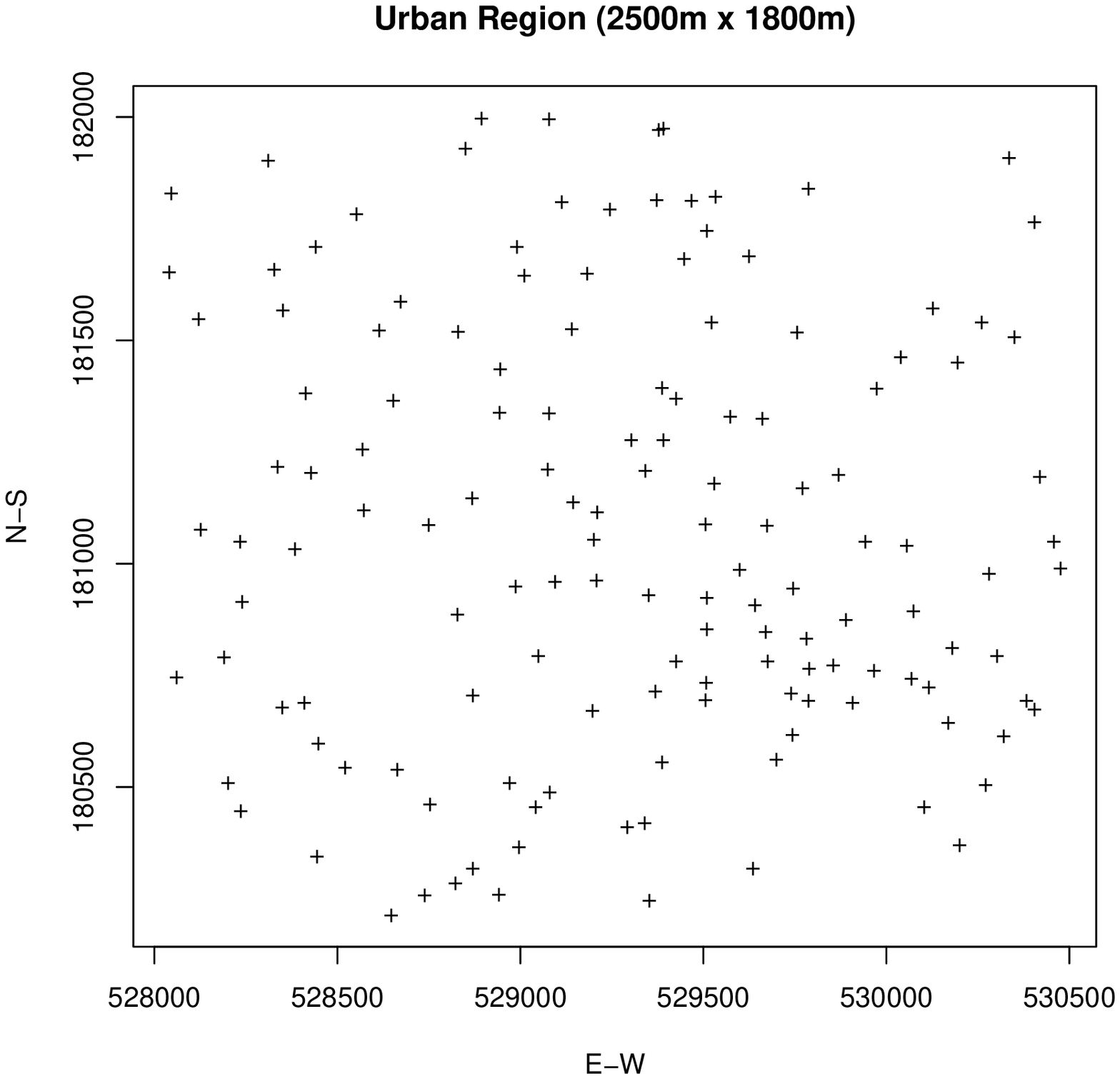}
    \caption{The locations of BSs in the urban region.}
    \label{fig:Urban_region}
\end{figure}

\begin{figure}
    \centering
    \includegraphics[width=0.5\textwidth]{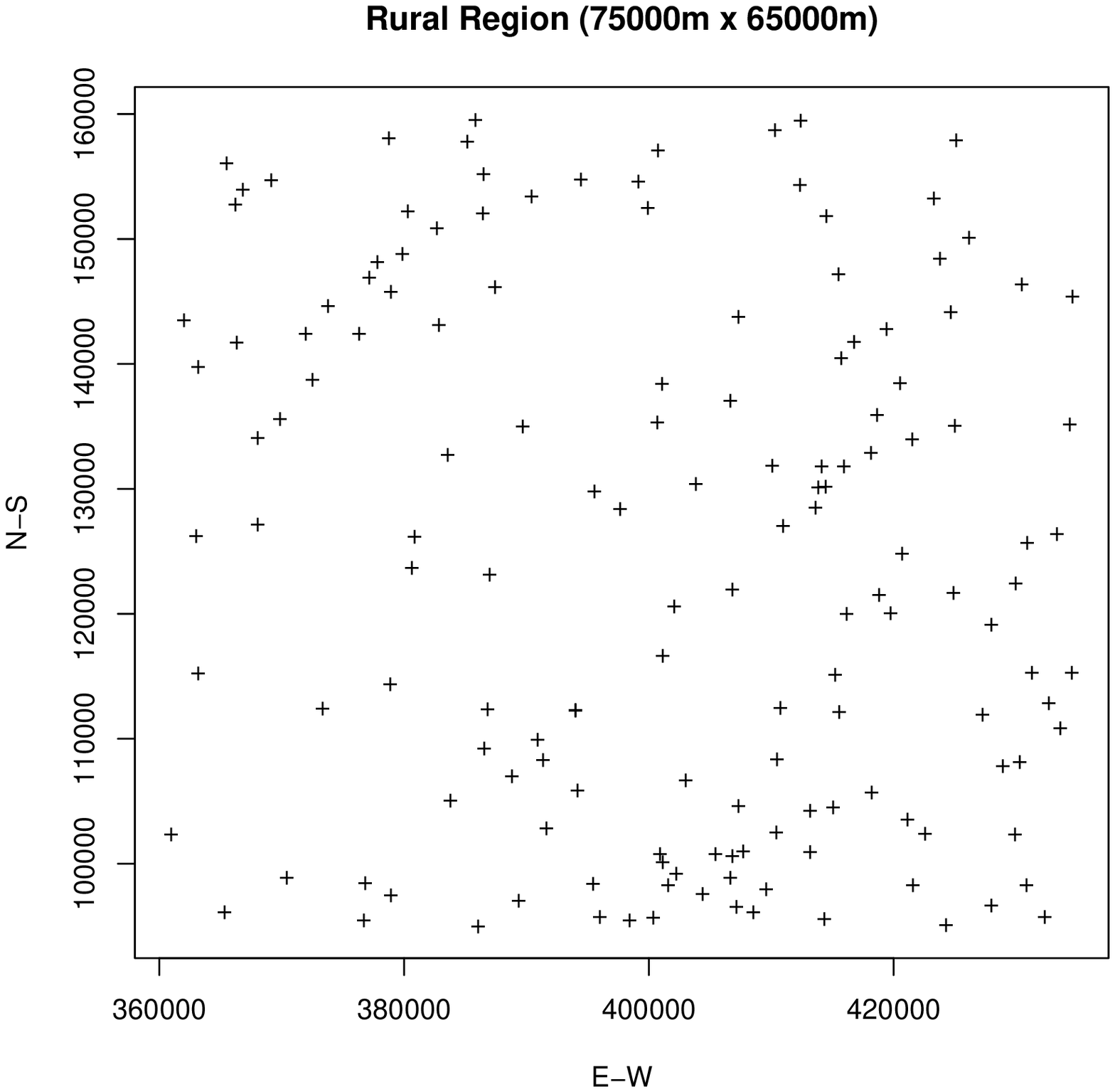}
    \caption{The locations of BSs in the rural region.}
    \label{fig:Rural_region}
\end{figure}

\fi
\begin{table*}[t]
\setlength{\abovecaptionskip}{0pt}
\setlength{\belowcaptionskip}{10pt}
  \caption{Details of the two fitted Region} \label{table 1}
  \centering
\begin{tabular}{|c|c|c|c|c|}
\hline
& Operator  & Area ($m\times m$)  & Number of BSs & Estimated Density\\
\hline
Urban Region & Vodafone  &$2500\times 1800$     &142 & ${\rm 3.156e\!-\!5m^{-2}}$\\
\hline
Rural Region & Vodafone  & $75000\times 65000$  &149 & ${\rm 3.056e\!-\!8m^{-2}}$\\
\hline
\end{tabular}
\end{table*}
\begin{figure}
  \centering
  \subfigure[Urban region.]{
    \label{fig:J function:a}
    \includegraphics[width=0.45\textwidth]{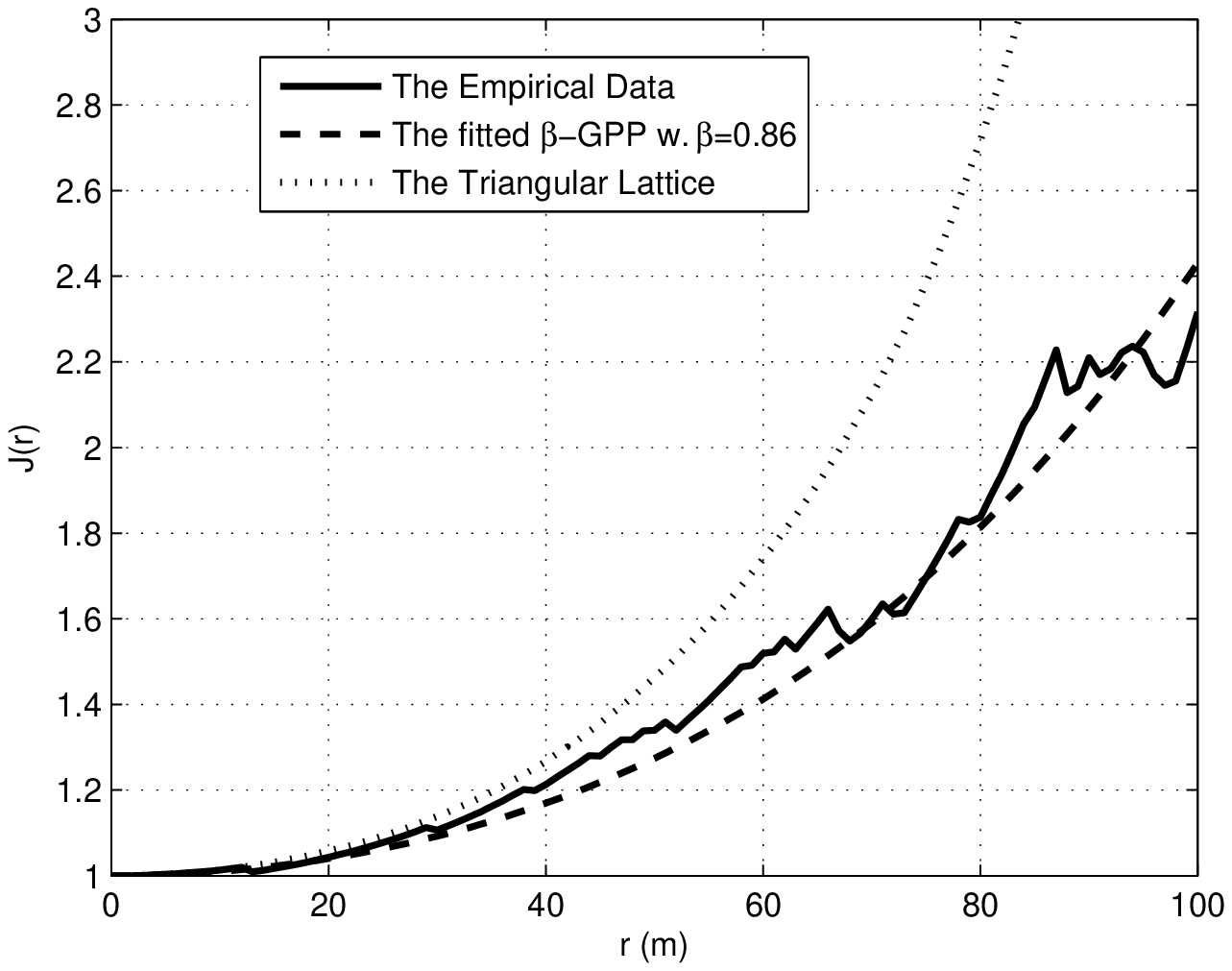}}
  \subfigure[Rural region.]{
    \label{fig:J function:b}
    \includegraphics[width=0.45\textwidth]{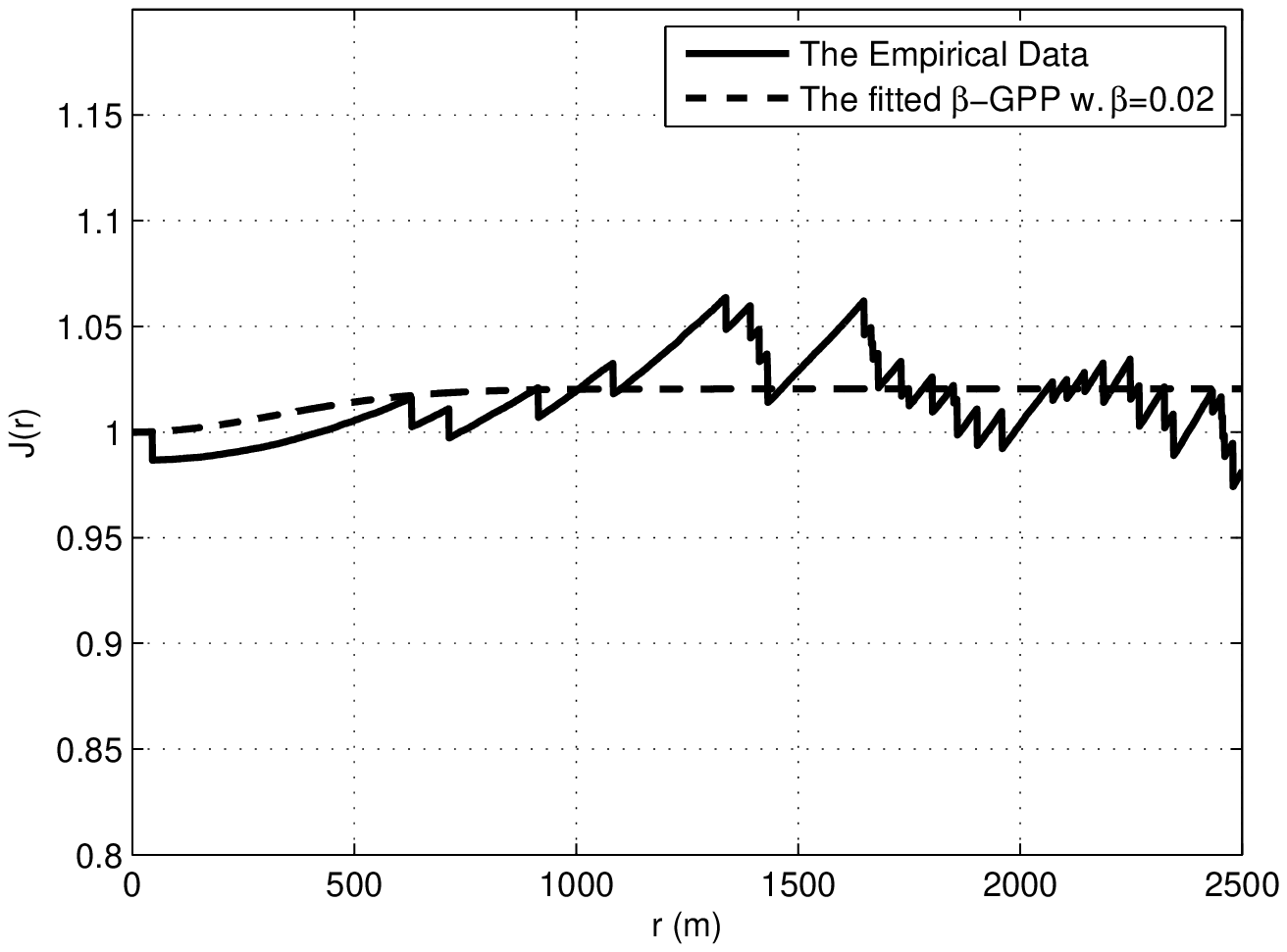}}
  \caption{The $J$ function and the corresponding fits.}
  \label{fig:J function} 
\end{figure}
\subsection{Fitting Result for J Function}
In this part we present the fitting results for both urban and rural regions using the $J$ function as the metric, shown in Figure \ref{fig:J function}. For comparison, we also include the $J$ function of the triangular lattice. To obtain the empirical $J$ function, we have to find the empirical $F$ and $G$ functions first. In the simulation of the $F$ function, we focus on the central part [$\frac{1}{2}\mathrm{length}\times\frac{1}{2}\mathrm{width}$] of the fitting region to mitigate the boundary effect, and the $F$ function is computed based on $10^5$ uniformly chosen locations at random. The $G$ function is obtained by calculating the distance from each point in the fitting region to its nearest neighbor.
From the figure, we can observe that both the urban and the rural regions are far less regular than the classical lattice model. Specifically, the empirical $J$ function of the rural region tends to that of the PPP and the one of the urban region matches well with that of the $\beta$-GPP with $\beta=0.86$. Thus, the $\beta$-GPP is more suitable and accurate for modeling these deployments.
\begin{figure}
  \centering
  \subfigure[Urban region.]{
    \label{fig:coverage probability:a}
    \includegraphics[width=0.45\textwidth]{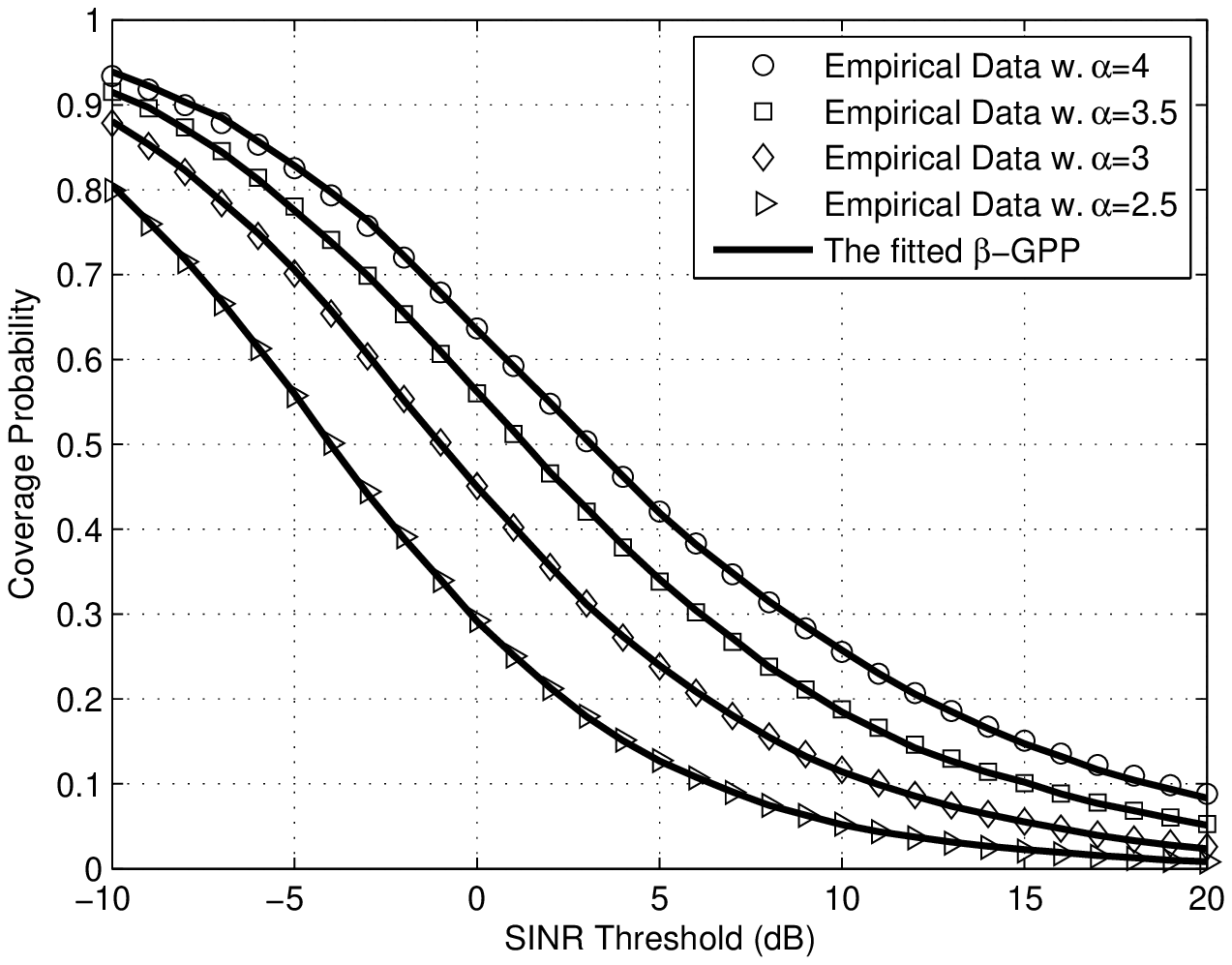}}
  \subfigure[Rural region.]{
    \label{fig:coverage probability:b}
    \includegraphics[width=0.45\textwidth]{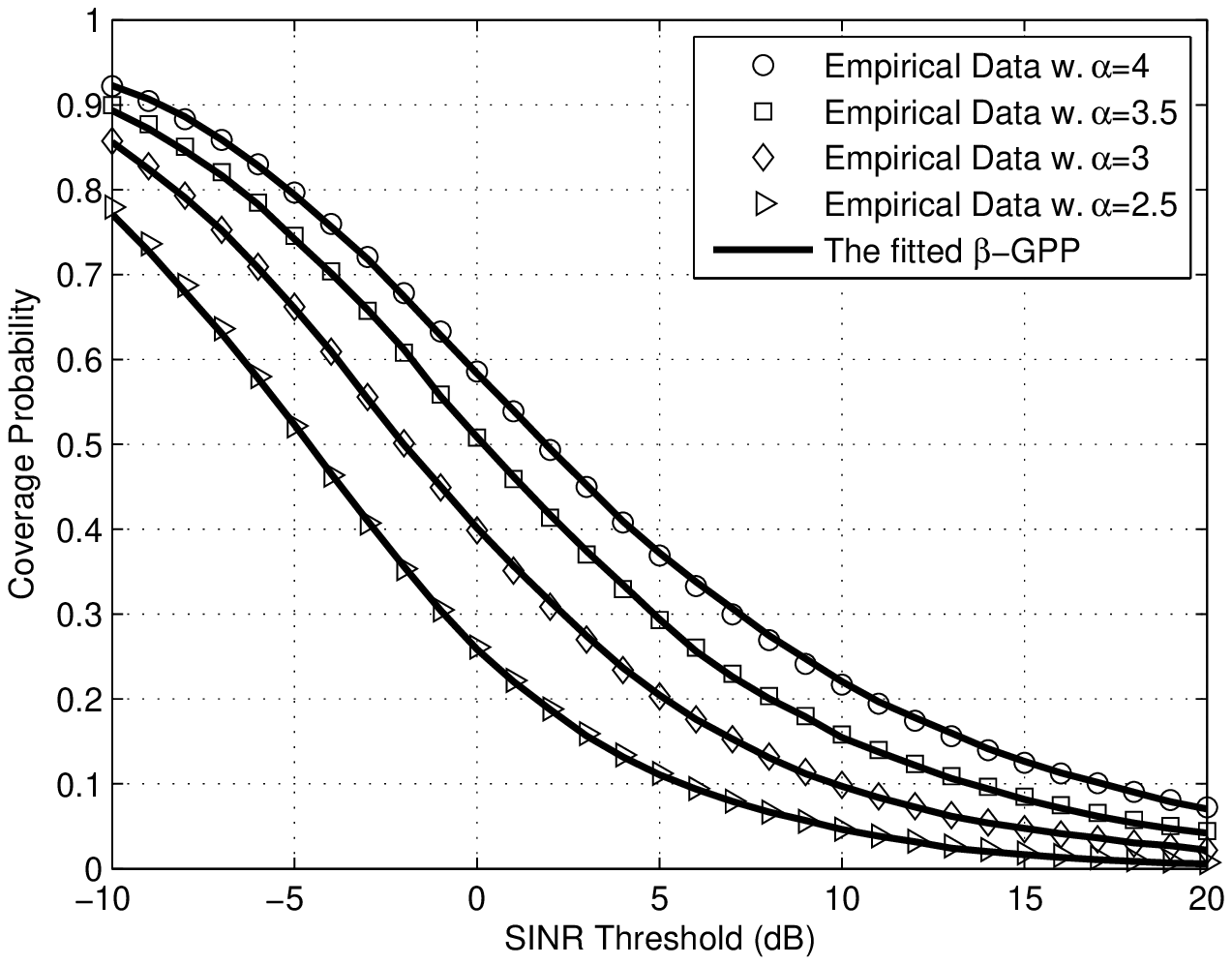}}
  \caption{The coverage probability for two regions with different $\alpha$.}
  \label{fig:coverage probability} 
\end{figure}
\begin{table}[t]
\setlength{\abovecaptionskip}{0pt}
\setlength{\belowcaptionskip}{10pt}
  \caption{Fitting results for the coverage probability} \label{table 2}
  \centering
\begin{tabular}{|c|c|c|c|c|}
\hline
 & $\alpha$=4 & $\alpha$=3.5 &$\alpha$=3 &$\alpha$=2.5 \\
\hline
Urban $\beta$ & 0.925 & 0.900 &0.975&0.925\\
\hline
Rural $\beta$ & 0.225 & 0.200 &0.225&0.375\\
\hline
\end{tabular}
\end{table}
\subsection{Fitting Result for Coverage Probability}
In this part we give the fitting results for both rural and urban regions with different path loss exponent $\alpha$ using the coverage probability as the metric, shown in Figure \ref{fig:coverage probability}. In the finite region, the empirical coverage probability $P_{c}(\theta)$ can be estimated by determining the covered area fraction for $\mathrm{SINR}>\theta$. In the following simulations, $P_{c}(\theta)$ is obtained by evaluating $10^5$ values of the SINR based on the $10^5$ randomly chosen locations. In order to mitigate the boundary effect, we only use the central [$\frac{1}{2}\mathrm{length}\times\frac{1}{2}\mathrm{width}$] rectangle of the fitting region. It should be noted that for smaller $\alpha$, most of the interference comes from far-away interferers but in the empirical data, they are not present since the fitting region is finite. Therefore, we add an analytical interference term that represents the mean interference obtained from interferers outside the fitting region. To obtain the curve for the $\beta$-GPP, we focus on the coverage probability of the user located at the origin. Instead of using the result in Theorem \ref{cov}, we exploit the exceptional simplicity of the distribution of the squared moduli (distances) of the points and simulate them using the gamma distribution method (see Proposition \ref{4}). We also evaluate $10^5$ values of the SINR based on the $10^5$ realizations of the squared moduli of the points in $\beta$-GPP. Table \ref{table 2} gives the fitting results of $\beta$ for different $\alpha$ in urban and rural regions, which reveals that the urban deployments are fairly regular ($\beta$ is close to 1) but not more regular than the $1$-GPP while the rural ones are quite irregular. As the figure shows, the curve of the fitted $\beta$-GPP and the curve of the point set match extremely well. Therefore, we can conclude that the scaled $\beta$-GPP is capable of modeling actual cellular networks by tuning the parameter $\beta$.

\section{Conclusions}
\label{sec:conclusions}
In this paper, we proposed a wireless network model according to the thinned and re-scaled GPP, which is a repulsive point process. Based on this model, we derived the mean and variance of interference and analyzed their finiteness using two different approaches: one is based on the Palm measure of the $\beta$-GPP and the other is based on the reduced second moment density obtained from the kernel of the $\beta$-GPP. For a bounded path-loss law, both the mean and variance of interference are finite when $\alpha>2$, while for an unbounded path-loss law, the mean interference and the variance are finite on different intervals of $\alpha$. Using the analytically obtained mean and variance, we also provided approximations of the interference distribution using three known PDFs, i.e., the gamma distribution, the inverse Gaussian distribution and the inverse gamma distribution. Through comparison with the Monte-Carlo simulations, we observed that, among the three PDFs, the inverse gamma distribution performs the best fit for $\alpha=3$ and the inverse Gaussian distribution provides the best fit for $\alpha=4$.

Besides, we derived a computable integral representation for the coverage probability of the typical mobile user. To demonstrate the accuracy and practicability of the $\beta$-GPP model for wireless networks, we fitted the point process model to publicly available base station data by carefully adjusting the parameter $\beta$, which controls the degree of repulsion between the points. Through fitting by minimizing the vertical average squared error, we found that the fitted $\beta$-GPP has nearly the same coverage probability as the given point set, and thus, in terms of coverage probability, is an accurate model for real deployments of the base stations. Through the fitting result of $\beta$, we found that the urban deployments are fairly regular but not more regular than the $1$-GPP, while the rural ones are quite irregular, i.e., they have smaller $\beta$ than the urban ones, which demonstrates that the $\beta$-GPP is capable of modeling all the actual cellular networks through tuning the parameter $\beta$.

Compared to the PPP, the $\beta$-GPP better captures the spatial distribution of the nodes in a real network deployment; compared to the MHCP or the Strauss process, the $\beta$-GPP provides more theoretical insights. Therefore, our work highlights the key role of the $\beta$-GPP for wireless networks with repulsion since it balances the accuracy, tractability and practicability tradeoffs quite well.

\section*{Acknowledgment}
The authors would like to thank Anjin Guo for providing the data sets used for the fitting in Section V.

\begin{appendices}
\section{Proof of Theorem \ref{1}}
\label{sect:Appendix A}
For the mean interference, we have
\ifCLASSOPTIONonecolumn
\begin{eqnarray}
\mathbb{E}_o^!(I)\!\!\!\!&=&\!\!\!\!\sum_{k=2}^{\infty}\mathbb{E}\left((\max\{r_0^2,Q_k\})^{-\alpha/2}T_k\right)\nonumber \\
\!\!\!\!&=&\!\!\!\!\beta\sum_{k=2}^{\infty}\int_0^{\infty}(\max\{r_0^2,q\})^{-\alpha/2}\frac{(c/\beta)^k}{\Gamma(k)}q^{k-1}e^{-\frac{c}{\beta}q}dq\nonumber \\
\!\!\!\!&=&\!\!\!\!c\int_0^{\infty}(\max\{r_0^2,q\})^{-\alpha/2}(1-e^{-\frac{c}{\beta}q})dq \nonumber \\
\!\!\!\!&=&\!\!\!\!cr_0^{2-\alpha}\frac{\alpha}{\alpha-2} + \beta r_0^{-\alpha}(e^{-\frac{c}{\beta}r_0^2}-1)-c^{\frac{\alpha}{2}}\beta^{1-\frac{\alpha}{2}}\Gamma\left(1-\frac{\alpha}{2},\frac{c}{\beta}r_0^2\right),\quad\alpha>2.
\end{eqnarray}
\else
\begin{eqnarray}
\mathbb{E}_o^!(I)\!\!\!\!&=&\!\!\!\!\sum_{k=2}^{\infty}\mathbb{E}\left((\max\{r_0^2,Q_k\})^{-\alpha/2}T_k\right)\nonumber \\
\!\!\!\!&=&\!\!\!\!\beta\sum_{k=2}^{\infty}\int_0^{\infty}(\max\{r_0^2,q\})^{-\alpha/2}\frac{(c/\beta)^k}{\Gamma(k)}q^{k-1}e^{-\frac{c}{\beta}q}dq\nonumber \\
\!\!\!\!&=&\!\!\!\!c\int_0^{\infty}(\max\{r_0^2,q\})^{-\alpha/2}(1-e^{-\frac{c}{\beta}q})dq \nonumber \\
\!\!\!\!&=&\!\!\!\!cr_0^{2-\alpha}\frac{\alpha}{\alpha-2} + \beta r_0^{-\alpha}(e^{-\frac{c}{\beta}r_0^2}-1) \nonumber \\
&&-c^{\frac{\alpha}{2}}\beta^{1-\frac{\alpha}{2}}\Gamma\left(1-\frac{\alpha}{2},\frac{c}{\beta}r_0^2\right),\quad\alpha>2.
\end{eqnarray}
\fi
When $c\rightarrow\infty$, $\beta r_0^{-\alpha}(e^{-\frac{c}{\beta}r_0^2}-1)\rightarrow-\beta r_0^{-\alpha}$ and
\ifCLASSOPTIONonecolumn
\begin{equation}
\lim\limits_{c\rightarrow\infty}c^{\frac{\alpha}{2}}\beta^{1-\frac{\alpha}{2}}\Gamma\left(1-\frac{\alpha}{2},\frac{c}{\beta}r_0^2\right)=\lim\limits_{c\rightarrow\infty}c^{\frac{\alpha}{2}}\beta^{1-\frac{\alpha}{2}}\int_{\frac{c}{\beta}r_0^2}^{\infty}t^{-\frac{\alpha}{2}}e^{-t}dt=0.
\end{equation}
\else
\begin{multline}
\lim\limits_{c\rightarrow\infty}c^{\frac{\alpha}{2}}\beta^{1-\frac{\alpha}{2}}\Gamma\left(1-\frac{\alpha}{2},\frac{c}{\beta}r_0^2\right) \\
=\lim\limits_{c\rightarrow\infty}c^{\frac{\alpha}{2}}\beta^{1-\frac{\alpha}{2}}\int_{\frac{c}{\beta}r_0^2}^{\infty}t^{-\frac{\alpha}{2}}e^{-t}dt=0.
\end{multline}
\fi
Thus, $\mathbb{E}_o^!(I)\rightarrow a+bc$ as $c\rightarrow\infty$, where $a=-\beta r_0^{-\alpha}$, $b=\frac{\alpha}{\alpha-2}r_0^{2-\alpha}$, i.e.,
the mean interference approaches an affine function as $c\rightarrow\infty$.

For the variance of the interference, we have
\ifCLASSOPTIONonecolumn
\begin{eqnarray}
V_o^!(I)\!\!\!\!&=&\!\!\!\!\sum\limits_{k=2}^{\infty}V_o^!\Big((\max\{r_0^2,Q_k\})^{\alpha/2}h_kT_k\Big) \nonumber \\
\!\!\!\!&=&\!\!\!\!\sum\limits_{k=2}^{\infty}\beta\mathbb{E}(h_k^2)\mathbb{E}\Big((\max\{r_0^2,Q_k\})^{-\alpha}\Big)-\beta^2\mathbb{E}^2\Big((\max\{r_0^2,Q_k\})^{-\alpha/2}\Big) \nonumber \\
\!\!\!\!&=&\!\!\!\!2c\int_0^{\infty}(\max\{r_0^2,q\})^{-\alpha}(1-e^{-\frac{c}{\beta}q})dq-\beta^2\sum\limits_{k=2}^{\infty}\left(\int_0^{\infty}(\max\{r_0^2,q\})^{-\alpha/2}\frac{(c/\beta)^k}{\Gamma(k)}q^{k-1}e^{-\frac{c}{\beta}q}dq\right)^2 \nonumber \\
\!\!\!\!&=&\!\!\!\!2\frac{c\alpha r_0^{2-2\alpha}}{\alpha-1}-2\beta r_0^{-2\alpha}(1-e^{-\frac{c}{\beta}r_0^2})-2c^{\alpha}\beta^{1-\alpha}\Gamma\left(1-\alpha,\frac{c}{\beta}r_0^2\right) \nonumber \\
&&-\beta^2\sum\limits_{k=2}^{\infty}\left(r_0^{-\alpha}\widetilde{\gamma}\Big(k,\frac{c}{\beta}r_0^2\Big)+\Big(\frac{c}{\beta}\Big)^{\frac{\alpha}{2}}\frac{\Gamma(k-\frac{\alpha}{2})}{\Gamma(k)}\right)^2,\quad\alpha>1.
\end{eqnarray}
\else
\begin{eqnarray}
&&V_o^!(I)=\sum\limits_{k=2}^{\infty}V_o^!\Big((\max\{r_0^2,Q_k\})^{\alpha/2}h_kT_k\Big) \nonumber \\
\!\!\!\!&=&\!\!\!\!\!\!\sum\limits_{k=2}^{\infty}\!\beta\mathbb{E}(h_k^2)\mathbb{E}\!\Big(\!(\max\!\{r_0^2,\!Q\!_k\})^{-\!\alpha}\!\Big)\!\!-\!\!\beta^2\mathbb{E}^2\!\Big(\!(\max\!\{r_0^2,\!Q\!_k\})^{-\!\frac{\alpha}{2}}\!\Big) \nonumber \\
\!\!\!\!&=&\!\!\!\!\!2c\int_0^{\infty}(\max\{r_0^2,q\})^{-\alpha}(1-e^{-\frac{c}{\beta}q})dq \nonumber \\
&&-\beta^2\!\sum\limits_{k=2}^{\infty}\!\left(\int_0^{\infty}\!\!(\max\{r_0^2,q\})^{-\frac{\alpha}{2}}\frac{(c/\beta)^k}{\Gamma(k)}q^{k-1}e^{-\frac{c}{\beta}q}dq\right)^2 \nonumber \\
\!\!\!\!&=&\!\!\!\!\!2\frac{c\alpha r_0^{2\!-\!2\alpha}}{\alpha-1}\!-\!2\beta\frac{1\!-\!e^{-\!\frac{c}{\beta}r_0^2}}{r_0^{2\alpha}}\!-\!2c^{\alpha}\!\beta^{1\!-\!\alpha}\Gamma\!\left(\!1\!-\!\alpha,\frac{c}{\beta}r_0^2\!\right) \nonumber \\
&&-\beta^2\!\sum\limits_{k=2}^{\infty}\!\left(\!r_0^{-\alpha}\widetilde{\gamma}\Big(k,\!\frac{c}{\beta}r_0^2\Big)\!\!+\!\!\Big(\frac{c}{\beta}\Big)^{\frac{\alpha}{2}}\!\frac{\Gamma(k\!-\!\frac{\alpha}{2})}{\Gamma(k)}\!\right)^2,~\alpha>1. \nonumber \\
\end{eqnarray}
\fi

\section{Proof of Theorem \ref{cov}}
\label{sect:Appendix C}
\ifCLASSOPTIONonecolumn
\begin{eqnarray}
\!\!\!\!\!\!\!\!\!\!&&\mathbb{P}(\mathrm{SINR}>\theta)=\sum\limits_{i=1}^{\infty}\mathbb{P}(\mathrm{SINR}>\theta,B_o=i) \nonumber \\
\!\!\!\!&=&\!\!\!\!\sum\limits_{i=1}^{\infty}\mathbb{P}(\mathrm{SINR}>\theta,B_o=i\mid T_i=1)\mathbb{P}(T_i=1) \nonumber \\
\!\!\!\!&=&\!\!\!\!\sum_{i=1}^{\infty}\beta \mathbb{P}\left\{h_i>\frac{\theta\left(\sigma^2+\sum\limits_{k\in\mathbb{N}\setminus\{i\}}h_k Q_k^{-\alpha/2}T_k\right)}{Q_i^{-\alpha/2}}, B_o=i\right\} \nonumber \\
\!\!\!\!&=&\!\!\!\!\sum_{i=1}^{\infty}\beta\mathbb{E}\left\{e^{-\mu\theta\sigma^2Q_i^{\alpha/2}}\exp\left(-\mu\theta \sum\limits_{k\in\mathbb{N}\setminus\{i\}} h_k (Q_i/Q_k)^{\alpha/2}T_k\textbf{1}_{\{Q_i\leq Q_k\}}\right)\right\}\nonumber \\
\!\!\!\!&=&\!\!\!\!\sum_{i=1}^{\infty}\beta\mathbb{E}\left\{e^{-\mu\theta\sigma^2Q_i^{\alpha/2}}\prod\limits_{k\in\mathbb{N}\setminus\{i\}}\left(\beta\exp\left(-\mu\theta h_k (Q_i/Q_k)^{\alpha/2}T_{\{Q_i\leq Q_k\}}\right)+1-\beta\right)\right\}\nonumber \\
\!\!\!\!&=&\!\!\!\!\sum_{i=1}^{\infty}\beta\mathbb{E}\left\{e^{-\mu\theta\sigma^2Q_i^{\alpha/2}}\prod\limits_{k\in\mathbb{N}\setminus\{i\}}\left( \frac{\beta}{1+\theta (Q_i/Q_k)^{\alpha/2}T_{\{Q_i\leq Q_k\}}}+1-\beta\right)\right\}\nonumber \\
\!\!\!\!&=&\!\!\!\!\sum_{i=1}^{\infty}\beta\!\!\int_0^{\infty}\!\frac{(c/\beta)^i}{\Gamma(i)}u^{i-1}e^{-cu/\beta}e^{-\mu\theta u^{\frac{\alpha}{2}}\sigma^2}\!\!\!\int_{u}^{\infty}\!\frac{(c/\beta)^k}{\Gamma(k)}y^{k-1}e^{-cy/\beta}\!\!\!\prod\limits_{k\in\mathbb{N}\setminus\{i\}}\!\!\!\left( \frac{\beta}{1\!+\!\theta (u/y)^{\alpha/2}}\!+\!1\!-\!\beta\right) dydu
\nonumber \\
\!\!\!\!&=&\!\!\!\!\sum_{i=1}^{\infty}\beta\!\!\int_0^{\infty}\!\frac{s^{i-1}}{\Gamma(i)}e^{-s}e^{-\mu\theta\sigma^2(\frac{\beta s}{c})^{\alpha/2}}\!\!\!\!\!\prod\limits_{k\in\mathbb{N}\setminus\{i\}}\!\int_{s}^{\infty}\!\frac{v^{k-1}}{\Gamma(k)}e^{-v}\!\!\left( \frac{\beta}{1\!+\!\theta (s/v)^{\alpha/2}}\!+\!1\!-\!\beta\right) dvds \nonumber \\
\!\!\!\!&=&\!\!\!\!\beta\!\!\int_0^{\infty}\!e^{-s}e^{-\mu\theta\sigma^2(\frac{\beta s}{c})^{\alpha/2}}M(\theta,s,\alpha,\beta)S(\theta,s,\alpha,\beta)ds, \nonumber \\
\end{eqnarray}
\else
\begin{eqnarray}
\!\!\!\!\!\!\!\!\!\!&&\mathbb{P}(\mathrm{SINR}>\theta)=\sum\limits_{i=1}^{\infty}\mathbb{P}(\mathrm{SINR}>\theta,B_o=i) \nonumber \\
\!\!\!\!&=&\!\!\!\!\sum\limits_{i=1}^{\infty}\mathbb{P}(\mathrm{SINR}>\theta,B_o=i\mid T_i=1)\mathbb{P}(T_i=1) \nonumber \\
\!\!\!\!&=&\!\!\!\!\sum_{i=1}^{\infty}\beta \mathbb{P}\!\left\{\!h_i\!>\!\frac{\theta\left(\sigma^2\!+\!\!\sum\limits_{k\in\mathbb{N}\setminus\{i\}}\!\!h_k Q_k^{-\alpha/2}T_k\right)}{Q_i^{-\alpha/2}}, B_o\!=\!i\!\right\} \nonumber \\
\!\!\!\!&=&\!\!\!\!\sum_{i=1}^{\infty}\beta\mathbb{E}\left\{e^{-\mu\theta\sigma^2Q_i^{\alpha/2}}\times\right. \nonumber \\
&&\left.\exp\Bigg(-\mu\theta \sum\limits_{k\in\mathbb{N}\setminus\{i\}} h_k (Q_i/Q_k)^{\alpha/2}T_kT_{\{Q_i\leq Q_k\}}\Bigg)\right\}\nonumber \\
\!\!\!\!&=&\!\!\!\!\sum_{i=1}^{\infty}\beta\mathbb{E}\left\{e^{-\mu\theta\sigma^2Q_i^{\alpha/2}}\times\right. \nonumber \\
&&\left.\!\!\!\!\!\prod\limits_{k\in\mathbb{N}\setminus\{i\}}\!\!\left(\beta\!\exp\!\left(-\mu\theta h_k (Q_i/Q_k)^{\alpha/2}\textbf{1}_{\{Q_i\leq Q_k\}}\right)\!+\!1\!-\!\beta\right)\right\}\nonumber \\
\!\!\!\!&=&\!\!\!\!\sum_{i=1}^{\infty}\beta\mathbb{E}\left\{e^{-\mu\theta\sigma^2Q_i^{\alpha/2}}\times\right. \nonumber \\
&&\left.\prod\limits_{k\in\mathbb{N}\setminus\{i\}}\left( \frac{\beta}{1+\theta (Q_i/Q_k)^{\alpha/2}T_{\{Q_i\leq Q_k\}}}+1-\beta\right)\right\}\nonumber \\
\!\!\!\!&=&\!\!\!\!\sum_{i=1}^{\infty}\beta\!\!\int_0^{\infty}\!\frac{(c/\beta)^i}{\Gamma(i)}u^{i-1}e^{-cu/\beta}e^{-\mu\theta u^{\frac{\alpha}{2}}\sigma^2}\times\nonumber \\
&&\!\!\!\!\!\!\int_{u}^{\infty}\!\!\frac{(c/\beta)^k}{\Gamma(k)}y^{k\!-\!1}\!e^{-\frac{c}{\beta}y}\!\!\!\!\!\prod\limits_{k\in\mathbb{N}\setminus\{i\}}\!\!\!\left(\! \frac{\beta}{1\!+\!\theta (u/y)^{\alpha/2}}\!+\!1\!-\!\beta\!\right)\!dydu
\nonumber \\
\!\!\!\!&=&\!\!\!\!\sum_{i=1}^{\infty}\beta\!\!\int_0^{\infty}\!\frac{s^{i-1}}{\Gamma(i)}e^{-s}e^{-\mu\theta\sigma^2(\frac{\beta s}{c})^{\alpha/2}}\times \nonumber \\
&&\prod\limits_{k\in\mathbb{N}\setminus\{i\}}\!\int_{s}^{\infty}\!\frac{v^{k-1}}{\Gamma(k)}e^{-v}\!\!\left( \frac{\beta}{1\!+\!\theta (s/v)^{\alpha/2}}\!+\!1\!-\!\beta\right) dvds \nonumber \\
\!\!\!\!&=&\!\!\!\!\beta\!\!\int_0^{\infty}\!e^{-s}e^{-\mu\theta\sigma^2(\frac{\beta s}{c})^{\alpha/2}}M(\theta,s,\alpha,\beta)S(\theta,s,\alpha,\beta)ds, \nonumber \\
\end{eqnarray}
\fi
where $\textbf{1}_{A}$ denotes the indicator for set $A$.
\end{appendices}

\end{document}